\documentclass[aps,prd,10pt,notitlepage,nofootinbib,superscriptaddress,showkeys,showpacs]{revtex4-1}
\usepackage{hyperref}
\usepackage{amsmath,amssymb,amsthm,latexsym,bbm,calc}
\usepackage[english]{babel}
\usepackage{graphicx,color}
\usepackage{xspace}
\usepackage{graphicx}
\usepackage{pifont,dsfont}
\usepackage{marvosym}
\usepackage{slashed}
\usepackage{subfigure}

\newcommand{\be}{\begin{equation}}
\newcommand{\ee}{\end{equation}}
\newcommand{\beq}{\begin{equation}}
\newcommand{\eeq}{\end{equation}}
\newcommand{\bes}{\begin{eqnarray}}
\newcommand{\ees}{\end{eqnarray}}
\newcommand{\bqa}{\begin{eqnarray}}
\newcommand{\eqa}{\end{eqnarray}}
\newcommand{\bea}{\begin{eqnarray}}
\newcommand{\eea}{\end{eqnarray}}
\newcommand{\tabl} [2] {\begin{array} {#1} #2 \end{array}}

\newcommand{\N}{\mathds{N}}

\newcommand{\R}{\mathds{R}}
\newcommand{\C}{\mathds{C}}

\newcommand{\nn}{\nonumber}

\renewcommand{\N}{\mathbbm{N}}
\renewcommand{\R}{\mathbbm{R}}
\renewcommand{\C}{\mathbbm{C}}

\newcommand{\su}{\mathfrak{su}}

\newcommand{\SB}{\mathrm{SB}(2,\C)}
\newcommand{\ISO}{\mathrm{ISO}}

\newcommand{\cU}{{\mathcal U}}
\def\UQ{{\cU_{q}(\su(2))}}
\newcommand{\f}{\frac}
\newcommand{\cR}{{\mathcal R}}

\newtheorem{lemma}{Lemma}

\newcommand{\cD}{{\cal D}}

\newcommand{\cH}{{\cal H}}

\newcommand{\cO}{{\cal O}}

\newcommand{\la}{\langle}
\newcommand{\ra}{\rangle}

\def\tell{\tilde{\ell}}
\def\tu{\tilde{u}}
\def\one{{\bf 1}}

\def\cop{{\Delta}}
\def\ot{{\,\otimes\,}}
\def\dr{{\rightarrow}}
\def\mone{{^{-1}}}
\newcommand{\cA}{{\cal A}}
\def\veps{{\varepsilon}}
\def\ov{\overline}
\def\hell{\widehat{\ell}}
\def\psir{\psi_{\cR}}
\def\bt{{\bf t}}
\def\demi{{\frac{1}{2}}}
\newcommand{\cL}{{\cal L}}
\newcommand{\act}{{\triangleright}}

\DeclareMathOperator{\tr}{tr}

\DeclareMathOperator{\SU}{SU}
\DeclareMathOperator{\SL}{SL}
\DeclareMathOperator{\U}{U}

\DeclareMathOperator{\Inv}{Inv}

\begin{document}
%\title{\Large \bf A Hamiltonian operator for Loop Quantum Gravity with a non-vanishing cosmological constant}
\title{\Large \bf Towards the Turaev-Viro amplitudes from a Hamiltonian constraint}

\author{{\bf Valentin Bonzom}}\email{bonzom@lipn.univ-paris13.fr}
\affiliation{LIPN, UMR CNRS 7030, Institut Galil\'ee, Universit\'e Paris 13,
99, avenue Jean-Baptiste Cl\'ement, 93430 Villetaneuse, France, EU}

\author{{\bf Ma\"it\'e Dupuis}}\email{m2dupuis@uwaterloo.ca}
\affiliation{Department of Applied Mathematics, University of Waterloo, Waterloo, Ontario, Canada}

\author{{\bf Florian Girelli}}\email{fgirelli@uwaterloo.ca}
\affiliation{Department of Applied Mathematics, University of Waterloo, Waterloo, Ontario, Canada}

\date{\small\today}

\begin{abstract}
\noindent
3D Loop Quantum Gravity with a vanishing cosmological constant can be related to the quantization of the $\SU(2)$ BF theory discretized on a lattice. At the classical level, this discrete model characterizes discrete flat geometries and its phase space is built from $T^\ast \textrm{SU}(2)$. In a recent paper \cite{HyperbolicPhaseSpace}, this discrete model was deformed using the Poisson-Lie group formalism and was shown to  characterize discrete hyperbolic geometries while being still topological. Hence, it is a good candidate to describe the discretization of  $\textrm{SU}(2)$ BF theory with a (negative) cosmological constant. We proceed here to the quantization of this model. At the kinematical level, the Hilbert space is spanned by spin networks built on $\mathcal{U}_{q}(\mathfrak{su}(2))$ (with $q$ real). In particular, the quantization of the discretized Gauss constraint leads naturally to $\mathcal{U}_{q}(\mathfrak{su}(2))$ intertwiners. We also quantize the Hamiltonian constraint on a face of degree 3 and show that physical states are proportional to the quantum 6j-symbol. This suggests that the Turaev-Viro amplitude with $q$ real is a solution of the quantum Hamiltonian. This model is therefore a natural candidate to describe 3D loop quantum gravity with a (negative) cosmological constant.

%argued to be a candidate for a discretization  it was shown that a classical deformation of $T^\ast \SU(2)$ as phase space with a set of first class constraints generalizing this flat case can still make the model topological. But the obtained deformed model now characterizes discrete hyperbolic geometries. In this paper, we quantize this model. We focus more particularly on the Hamiltonian constraint and solve it.
%
%3D gravity can be quantized using different techniques such as Chern-Simons, spinfoam and loop quantum gravity. When the cosmological constant $\Lambda$ is zero, these different approaches have been showed to be equivalent.  When $\Lambda\neq0$, in the

\end{abstract}

\medskip

%\noindent  Pacs numbers: 02.10.Ox, 04.60.Gw, 05.40-a

\maketitle

%%%%%%%%%%%%%%%%%%%%%%%
\section{Introduction}
%%%%%%%%%%%%%%%%%%%%%%%

3D gravity has been an important testing ground in the aim of quantizing 4D gravity. It has provided quite a lot of ideas that are used today in loop quantum gravity (LQG). One of them is that some degrees of freedom of LQG can be understood as discretized gravity on a lattice, also known as models of discrete geometries. A key question that has yet to be answered is that of the continuum limit of such models and the fate of the diffeomorphism symmetry (see \cite{RestoringDiffBianca, CoarseGrainingSpinNets, CylindricalConsistence, PerfectActions, TimeEvoCoarseGraining} for many discussions on this issue).

Lots of effort have been put into 3D gravity because it is exactly solvable, topological and exact on the lattice such as on a triangulation (see \cite{ScalarProd3D, 3DHamiltonian, Spin1/2Hamiltonian}). The idea is similar to Regge calculus, an approximation to general relativity where the simplices are taken to be flat and curvature lies on the hinges. But in the 3D case with a vanishing cosmological constant, spacetime is flat and the discretization is exact. In particular the set of symmetries of the continuous reduces to a set of symmetries on the lattice which still makes the model topological \cite{CellularQuantization}. The fact that the symmetries survive discretization makes the physical vacuum of the theory quite interesting. It has even been very recently proposed as a new vacuum for LQG, as it is directly physical in the flat (topological) case and is much closer to the spirit of the spin foam quantization (spin foam models are based on modified topological transition amplitudes) \cite{NewVacuumBianca}.

However, in discrete geometry models for non-flat spacetimes, gauge symmetries (diffeomorphisms) are typically broken. This is even the case for 3D gravity with a non-zero cosmological constant and discretized {\it \`a la Regge}, \cite{BrokenSym}. It is nonetheless known that 3D gravity with a cosmological constant $\Lambda\neq0$ is a topological quantum theory (of the Chern-Simons type) \cite{3DGravityAsCS} and its partition function (for positive $\Lambda$) is known to be a quantum invariant, the Turaev-Viro (TV) invariant of 3-manifolds \cite{TV, WalkerTV, RobertsTV}. The TV invariant actually is an example of a spin foam model. It is moreover (formally\footnote{The TV invariant is finite, but it has no $q\to1$ limit in general.}, i.e. term by term in the state-sum model) obtained by taking the $q$-deformation ($q$ a root of unity) of all the spin foam data which come from the representation theory of $\SU(2)$.

All in all, the present situation strongly suggests that there should be a classical model of discrete, homogeneously curved, geometries whose quantization leads to TV transition amplitudes. Since the quantum TV model is based on the $q$-deformation of $\SU(2)$ and that the discrete model underlying LQG for flat geometry has $T^*\SU(2)$ as phase space, it was proposed in \cite{HyperbolicPhaseSpace} to use a classical deformation of $T^*\SU(2)$ as phase space and shown therein that a set of first class constraints generalizing the flat case still make the model topological. Moreover, the geometric content was interpreted as discrete hyperbolic geometries formed by gluing of hyperbolic triangles ($\Lambda<0$). This is a sort of gauge covariant, and canonical, extension of the Regge calculus proposed in \cite{CurvedRegge} that uses homogeneously curved simplices, and it is best suited for a generalization of LQG to $\Lambda<0$. In particular, this framework could be important to improve the discrete hypersurface deformation algebras found in \cite{DiscreteHDA} in the curved case.

%\begin{tabular}{ccccc}
%$BF$ \textrm{ action with  }$\Lambda\neq0 $& $\textcolor{red}{\stackrel{Discretization?}{\longrightarrow}}$& Discretized Hamiltonian model  & $\textcolor{blue}{\stackrel{Quantization}{\longrightarrow}}$&  $\left\{\begin{array}{l}
%\textrm{Turaev-Viro model}\\ \hspace{1cm} \textcolor{blue}{ \updownarrow} \\ \textrm{LQG based on } \UQ
%\end{array}\right.$
%\end{tabular}

\begin{figure}
\includegraphics[scale=.8]{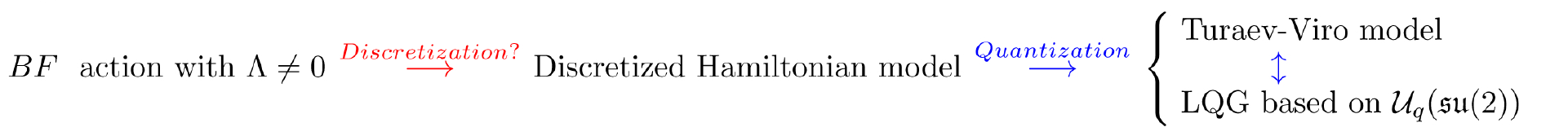}
\caption{The precise discretization scheme to describe the deformed model \cite{HyperbolicPhaseSpace} is still to be determined. We proceed here to the quantization of the  model \cite{HyperbolicPhaseSpace} and recover the Turaev-Viro amplitude as a solution of the Hamiltonian constraint, as well as intertwiners based on $\UQ$. This model is therefore a natural candidate to describe 3D Euclidian LQG with $\Lambda<0$ based on $\UQ$.}\label{plan}
\end{figure}

In the present paper we proceed to the quantization of the model of \cite{HyperbolicPhaseSpace}. We focus on the quantization of the deformed momenta which describe the intrinsic geometry of the embedded surface, following closely \cite{SternYakushin} though with different conventions. At the quantum level, the momenta become generators of $\UQ$. We impose the Gauss law at the quantum level, which is usually the first step in LQG, giving rise to spin network states. It is found that the Gauss law is solved by $q$-intertwiners. We also quantize and solve the Hamiltonian constraint on a face of degree three, following the techniques of \cite{3DHamiltonian} at $q=1$, and find the first hint that the transition amplitude associated to the 3-to-1 Pachner move will be given by the $q$-6j symbol (for real $q$), i.e. the building block of TV amplitudes. This is an important step, although a partial one, to justify the current use of quantum groups in spin foam models for (4D) gravity as a way to account for the cosmological constant. Besides being an important glimpse relating a classical model of discrete geometry to TV amplitudes (see also \cite{KauffmanBracketLQG, pranzetti}), this work also checks the robustness of the physical vacuum of the flat topological model with respect to introducing curvature (still in the context of a topological model -- but it would be hopeless in 4D gravity if it did not work in 3D with $\Lambda\neq0$).

%%%%%%%%%%%%%%%%%%%%%%%
\section{Classical setup} \label{sec:Classical}
%%%%%%%%%%%%%%%%%%%%%%%
We recall the main elements of the model introduced  in \cite{HyperbolicPhaseSpace}.  We consider a cell decomposition of an orientable 2D surface, and equip each edge of the decomposition with the Heisenberg double of $\SU(2)$, $\mathcal{D} = \SL(2,\C)$. It is formed by $\SU(2)$ and its dual, the group   $\SB$ which consists of lower triangular matrices of the form
\be
\ell = \begin{pmatrix} \lambda & 0 \\ z & \lambda^{-1} \end{pmatrix},
\ee
with $z\in \C$ and $\lambda$ a non-vanishing real number. For $u\in\SU(2)$, $\SL(2,\C)$ is found as products $G= \ell u$ corresponding to the left Iwasawa decomposition. The Poisson brackets are
\be \label{pb}
\begin{aligned}
&\{ \ell_1, \ell_2\} = -[r,\ell_1 \ell_2], && \{\ell_1, (\ell_2^\dagger)^{-1} \} = -[r, \ell_1 (\ell_2^\dagger)^{-1}], && \{(\ell_1^\dagger)^{-1}, (\ell_2^\dagger)^{-1} \} = -[r, (\ell_1^\dagger)^{-1} (\ell_2^\dagger)^{-1}] \\
&\{\ell_1, u_2\} = -\ell_1 r u_2, && \{(\ell_1^\dagger)^{-1}, u_2\} = -(\ell_1^\dagger)^{-1} r^\dagger u_2, && \{u_1, u_2\}= [r, u_1 u_2],
\end{aligned}
\ee
with the notations are $a_1 = a \otimes \mathbbm{I}, a_2 = \mathbbm{I} \otimes a$, and where the classical $r$-matrix is
\be
r = \frac{i\kappa}{4} \begin{pmatrix} 1 &0 &0 &0\\ 0 &-1 &0 &0 \\ 0 &4 &-1 &0 \\ 0 &0 &0 &1 \end{pmatrix}
\ee
An equivalent set of variables is found via the right Iwasawa decomposition, $G= \tu\tell$. To keep track of the equivalence between the two Iwasawa decompositions, it is convenient to picture the cell decomposition as a ribbon graph, with ribbon edges, ribbon vertices, and faces. Each ribbon edge can be seen as a box like in the Figure \ref{fig:RibbonEdge}, with group elements attached to the boundary of the box. The two ways to go from one corner to the opposite corner correspond to the constraint $\ell u = \tu \tell$.

\begin{figure}
\includegraphics[scale=.5]{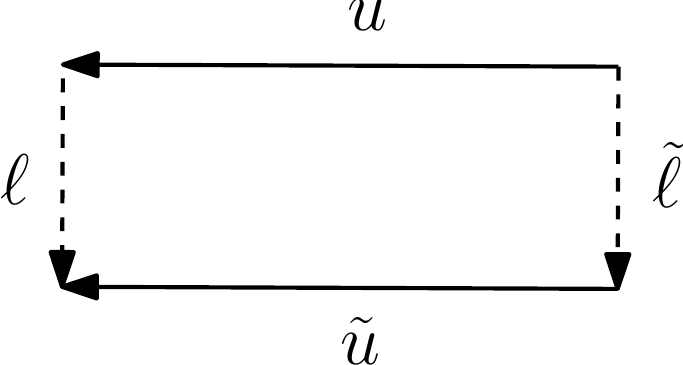}
\caption{An edge of the cell decomposition carries the variables $(\ell, u)$, or equivalently $(\tell, \tu)$. To represent the constraint $\ell u = \tu \tell$, it is convenient to thicken the edge and turn it to a ribbon edge  on which the constraint is the "commutativity" of the box.}\label{fig:RibbonEdge}
\end{figure}

This model can be seen as a deformation of the usual discretized model for $BF$ with $\Lambda=0$ \cite{HyperbolicPhaseSpace}. In this case, the Heisenberg double is given by the group $\ISO(3)$. The variables $\ell$ (and $\tell$) reduces to the generators of $\su(2)$, known as fluxes in the loop quantum gravity literature, while the variables $u$ (and $\tu$) are $\SU(2)$ holonomies. Following the construction of the flat case, we will call the variables $\ell$ and $\tell$ the momentum variables, whereas $u$ and $\tu$ are the configuration variables, or the holonomies.

The dynamics is given by a set of first class constraints. We have the Gauss constraints which restrict the momentum variables and the "flatness" constraints which restrict the configuration variables. We showed in \cite{HyperbolicPhaseSpace} how this set of constraints defines a topological model. Let us recall the construction of those constraints.

When ribbon edges are connected to a ribbon vertex $v$, the boundary of the vertex carries some matrices $\ell$ and $\tell$, with labels $1,\dotsc,N_v$. The Gauss law is defined by the requirement that the (oriented) product of those matrices around the vertex is trivial,
\be
\mathcal{G}_v \equiv \mathcal{L}_1 \dotsm \mathcal{L}_{N_v} = \one,\qquad\text{with $\mathcal{L}_i = \ell_i$ or $\tell^{-1}_i$}.
\ee
When the edge is oriented inward, $\cL_e = \ell_e$, and when it is outward, $\cL_e = \tell^{-1}_e$. These constraints generate local (vertex-based) $\SU(2)$ transformations on $\ell, u, \tell, \tu$ via the Poisson brackets \cite{HyperbolicPhaseSpace}. Importantly, those $\SU(2)$ transformations incorporate some ``braiding''. Indeed, consider for simplicity that the edges incident at $v$ are oriented inwards. Then, the constraint $\ell_1\ell_2\dotsm =\one$ generates a left transformation on the $\SL(2,\C)$ element $\ell_1 u_1 \to v\ell_1 u_1$, for $v\in\SU(2)$. The Iwasawa decomposition allows to find a unique $\ell_1'$ and a unique $u_1'$ such that $v \ell_1 u_1 = \ell_1' u_1'$. Defining $v'\in\SU(2)$ through $v\ell_1 = \ell_1'v'$, we find
\be
\ell_1'= v\ell_1 v'^{-1},\qquad u_1'=v'u_1.
\ee
with $\ell_1'\in\SB$. Then, on the edge 2, the Gauss law generates an $\SU(2)$ transformation with the ``braided'' rotation $v'$ instead of $v$,
\be
\ell_2'= v'\ell_2 v''^{-1},\qquad u_2'=v''u_2,
\ee
where $v''\in\SU(2)$ is defined through the Iwasawa decomposition applied to the $\SL(2,\C)$ equality $v'\ell_2 = \ell_2' v''$. It continues this way all around the vertex \cite{HyperbolicPhaseSpace}.

The second set of constraints generalizes the flatness constraint of regular BF theory and actually has the same form. This ``flatness'' constraint is the requirement that the product of $\SU(2)$ elements along the boundary of each face $f$ is trivial,
\be
\mathcal{C}_f \equiv \mathcal{U}_1 \dotsm \mathcal{U}_{N_f} = \one,\qquad\text{with $\mathcal{U}_i = u_i$ or $\tu^{-1}_i$},
\ee
These constraints generate (right and left) multiplication of $\ell, u$ by lower triangular matrices which can be viewed as a deformation of the $\ISO(3)$ ($\kappa=0$ case) translations \cite{HyperbolicPhaseSpace}. The whole set of constraints is easily visualized graphically, like in  Figure \ref{fig:RibbonGraph}.
\begin{figure}
\includegraphics[scale=.5]{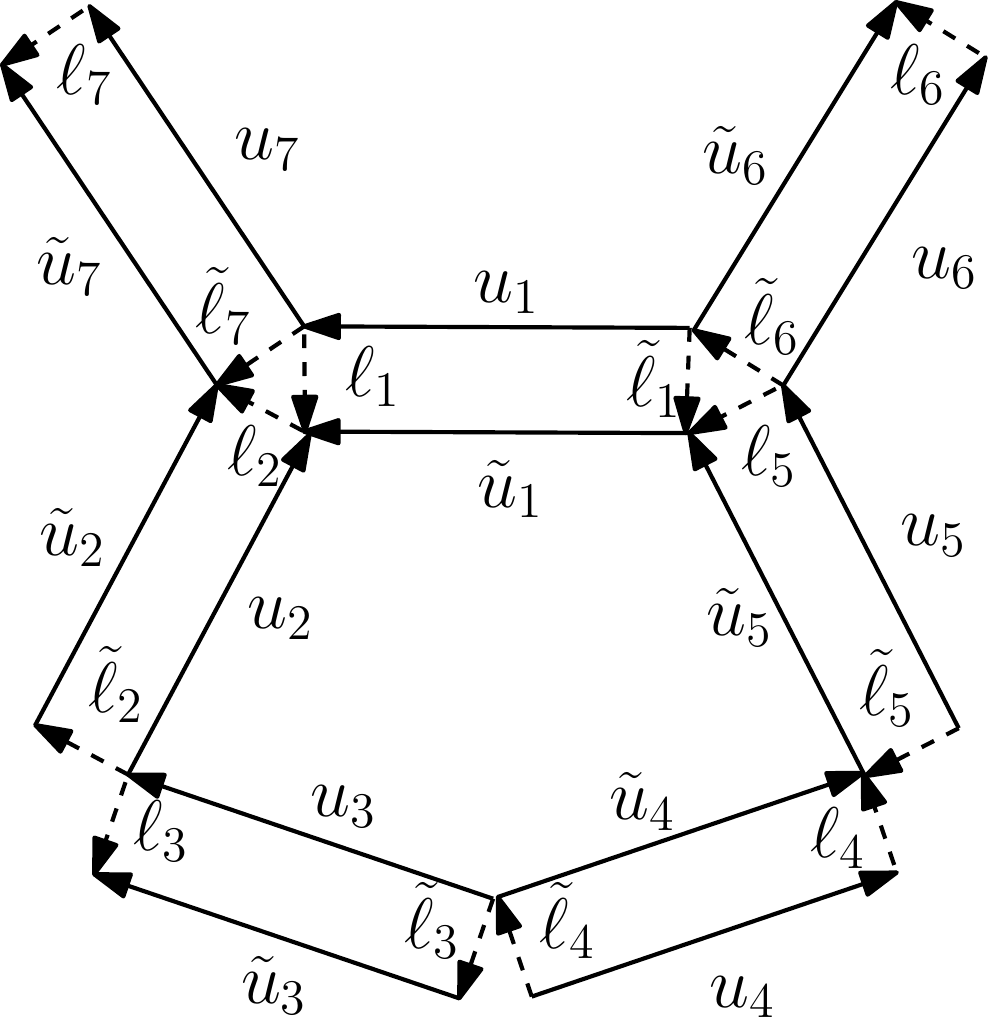}
\caption{This is an open portion of a ribbon graph. The edges 1, 2, 7 meet at a ribbon vertex with the constraint $\tell_7^{-1} \ell_2 \ell_1=\one$, and the edges 1, 5, 6 meet at another vertex with the constraint $\ell_5\tell_6^{-1}\tell_1^{-1}=\one$. The edges 1, 2, 3, 4, 5 close to form a face with the constraint $\tu_1^{-1} u_2 u_3 \tu_4^{-1}\tu_5^{-1}=\one$. }\label{fig:RibbonGraph}
\end{figure}

%%%%%%%%%%%%%%%%%%%%%%%
\section{Quantization of the Gauss law}
%%%%%%%%%%%%%%%%%%%%%%%

We would like to consider the quantization of the Gauss constraint, which is expressed in terms of the momentum variables $\ell_i, \, \tell_j$. Hence we need to quantize the algebra of functions on momentum space, which can be viewed here as the algebra $M(\SB)$ generated by the matrix elements of $\SB$. After recalling the Hopf algebra structure of  $M(\SB)$, we are going to recall how upon quantization, the algebra of quantum observables is equivalent to $\UQ$. We will see then how the classical Gauss law we have defined leads to $\UQ$ intertwiners, which are used in LQG based on $\UQ$ \cite{ours1}.

%%%%%%%%%%%%%
\subsection{Recovering $\UQ$}
%%%%%%%%%%%%%

As recalled in  \cite{HyperbolicPhaseSpace}, we can interpret the Heisenberg double $\SL(2,\C)$ as the phase space of a deformed top. The Gauss constraint can be interpreted as saying that the total (deformed) angular momentum of the tops meeting at this vertex ought to be zero.

Before discussing the quantization of the deformed momentum variables let us recall the key mathematical objects that are needed when quantizing a set of identical classical systems, in a simple known example.

\medskip

Consider the classical algebra of momentum observables given by $C^{\infty}(\R)$, spanned by monomials in $p$, and quantized in the algebra $\cA$. The momentum operator $\hat P$ is diagonal in the standard momentum basis of 1-particle states, $\hat P|p\ra= p|p\ra $. The mathematical structures required in the construction of the inverse momentum and of the total momentum for two particles %are encoded in the coalgebra structure of $\cA$, and 
are respectively given by the antipode and the coproduct.  (We could also consider the analogue of the zero momentum which would be given by the counit $\veps$ but it will not be used in the present article.)

\begin{itemize}
\item The coproduct $\cop:\cA\dr \cA\ot \cA$ dictates how to recover an observable for a 2-particle system from the knowledge of this observable on a 1-particle system. By induction, we can have the observable for an $n$-particle system.
For example, the total momentum of a 2-particle state $|p_1,p_2\ra$ is given by  $\cop \hat P = \hat P\ot \one + \one \ot \hat P$, so that $\cop \hat P|p_1,p_2\ra = (p_1+p_2)|p_1,p_2\ra$.

\item  The antipode $S:\cA\dr \cA$ provides in this case the notion of inverse momentum since $S(\hat P) = -\hat P$ and $S(\hat P)|p\ra= -p|p\ra $.

%\item The counit $\veps:\cA\dr \cA$ provides the notion of zero momentum since $\veps (\hat P) = 0$ and $\veps (\one) = \one$.

\end{itemize}
The objects, $\cop, \, S$ (together with the counit $\veps$) satisfy some compatibility relations, so that $\cA$ becomes a Hopf algebra. One property is that the coproduct is coassociative. In the previous example it is cocommutative but it is not in general true, as we will see in the case of interest for this paper. For a review of Hopf algebras we refer to \cite{majid}.

\medskip

We need now to generalize this construction to our case. We deal with the algebra of momentum observables $M(\SB)$ generated  by the matrix elements of $\SB$, which are $\lambda, \lambda\mone, z, \ov{z}$. The relevant features of the coalgebra of  $M(\SB)$ are given by the following coproduct $\Delta$ and antipode $S$ \cite{majid}
\bes \label{CoalgebraMSB}
&& \begin{array}{cccc}
\cop :&M(\SB) &\dr& M(\SB)\,\ot\, M(\SB) \\
&m_{ij}& \dr& \cop m_{ij} = \sum_k \,m_{ik} \,\ot\, m_{kj},
\end{array} \qquad
\begin{array}{cccc}
 S :&M(\SB) &\dr& M(\SB) \\
&m_{ij}& \dr& S(m_{ij})= {(m\mone)}_{ij}
\end{array}
%\\
%&& \begin{array}{cccc}
% \veps  :&M(\SB) &\dr& M(\SB) \\
%&m_{ij}& \dr& \veps(m_{ij})= {\delta}_{ij}
%\end{array}
\ees
We are now ready to quantize our momentum variables $\ell$ and $\ell^{\dagger}\mone$. We introduce $q= e^{\hbar\kappa}$ and the quantization rule \cite{SternYakushin}
\be
\lambda \,\dr\, K, \quad \lambda\mone \,\dr\,  K\mone, \quad z \,\dr\,  (q^{\frac12} - q^{-\frac12}) J_+, \quad \ov{z} \,\dr\,  -(q^{\frac12} - q^{-\frac12}) J_-,
\ee
where $K, J_\pm$ are operators. Hence we get the quantum matrices
\be \label{QuantumEll}
%\ell = \begin{pmatrix} \lambda & 0 \\ z & \lambda^{-1} \end{pmatrix} \longrightarrow
\widehat{\ell} = \begin{pmatrix} K & 0 \\ (q^{\frac12} - q^{-\frac12}) J_+ & K^{-1} \end{pmatrix},
\qquad \widehat{\ell}^{\dagger-1} = \begin{pmatrix} K^{-1} & -(q^{\frac12} - q^{-\frac12}) J_- \\ 0 & K \end{pmatrix}.
\ee
The commutation relations of the operators $J_\pm$ and $K$ are given by the quantization of the Poisson brackets \eqref{pb},
\be \label{RLL=LLR}
\cR\,\widehat{\ell}_1\,\widehat{\ell}_2 = \widehat{\ell}_2\,\widehat{\ell}_1\,\cR,\qquad
\cR\,\widehat{\ell}_1^{\dagger -1}\,\widehat{\ell}_2^{\dagger -1} = \widehat{\ell}_2^{\dagger -1}\,\widehat{\ell}_1^{\dagger -1}\,\cR, \qquad
\cR\,\widehat{\ell}_1\,\widehat{\ell}_2^{\dagger -1} = \widehat{\ell}_2^{\dagger -1}\,\widehat{\ell}_1\,\cR,
\ee
where the quantum $\cR$-matrix is
\be
\cR = \begin{pmatrix} q^{\frac14} & 0 & 0 & 0\\ 0 & q^{-\frac14} & 0 & 0\\ 0 & q^{-\frac14} (q^{\frac12} - q^{-\frac12}) & q^{-\frac14} & 0\\ 0 & 0 & 0 & q^{\frac14} \end{pmatrix}.
\ee
%These two operators do not commute because their entries are non-commutative.

The $\cR$-matrix expands as $\cR = \one -i\hbar\,r + \mathcal{O}(\hbar^2)$. Therefore the commutation relation gives at order $\hbar$, $[\widehat{\ell}_1,\widehat{\ell}_2] \simeq i\hbar [r,\widehat{\ell}_1\widehat{\ell}_2]$ and the classical algebra is recovered.

In components, the commutation relations \eqref{RLL=LLR} read
\be
K\,J_+\,K^{-1} = q^{\frac12}\,J_+,\qquad K\,J_-\,K^{-1} = q^{-\frac12}\,J_-,\qquad [J_+,J_-] = \frac{K^2 - K^{-2}}{q^{\frac12} - q^{-\frac12}}.
\ee
These are the commutation relations of $\mathcal{U}_q(\SU(2))$. We have recovered the well-known fact that the quantization of $M(\SB)$ leads to the algebra $\mathcal{U}_q(\SU(2))$.

Let us have a look at the coalgebra sector. First, the coproduct will provide us the total (deformed) quantum angular momentum and hence the quantum analogue of the Gauss law. Quantization simply promotes the coproduct of $M(\SB)$ in Equation \eqref{CoalgebraMSB} to an operator-valued coproduct,
\be
\cop\hell_{ij} = \sum_k \hell_{ik}\otimes \hell_{kj},
\ee
and similarly on $\hell^{\dagger -1}$. In matrix form, we get
\be
\begin{aligned}
&\cop \hell
= \begin{pmatrix} K\otimes K & 0\\ (q^{\frac12}-q^{-\frac12}) (J_{+} \otimes K + K^{-1}\otimes J_+) & K^{-1}\otimes K^{-1} \end{pmatrix},\\
\text{and}\qquad &\cop \hell^\dagger\mone
= \begin{pmatrix} K\otimes K &  -(q^{\frac12}-q^{-\frac12}) (J_{-} \otimes K + K^{-1}\otimes J_-)\\ 0 & K^{-1}\otimes K^{-1} \end{pmatrix}.
\end{aligned}
\ee
We here recognize the \emph{coproduct} $\cop$ of $\mathcal{U}_q(\SU(2))$,
\be
\Delta(K^{\pm1}) = K^{\pm1}\otimes K^{\pm1},\qquad \Delta(J_\pm) = J_\pm\otimes K + K^{-1}\otimes J_\pm.
\ee
%where $\Delta(J_-)$ can be read in $\Delta(\hat{\ell}^{\dagger -1})$.

The antipode of $M(\SB)$ simply is the matrix inversion. We require this to also hold quantum mechanically. Assuming a quantization map of the form
\be
\ell^{-1}\ \to\ S(\widehat{\ell}) = \begin{pmatrix} S(K) & 0\\ (q^{\frac12}-q^{-\frac12}) S(J_+) & S(K^{-1}) \end{pmatrix},
\ee
the equation determining $S(K^{\pm1}), S(J_\pm)$ is
\be
\widehat{\ell}\ S(\widehat{\ell}) = \begin{pmatrix} K & 0\\ (q^{\frac12}-q^{-\frac12}) J_+ & K^{-1} \end{pmatrix}  \begin{pmatrix} S(K) & 0\\ (q^{\frac12}-q^{-\frac12}) S(J_+) & S(K^{-1}) \end{pmatrix}  = \begin{pmatrix} \one &0\\ 0&\one\end{pmatrix},
\ee
and similarly for $ \widehat{\ell}^\dagger\mone$. The solution is
\be
S(K^{\pm1}) = K^{\mp1}, \quad S(J_\pm) = - K\,J_\pm\,K^{-1} = -q^{\pm\frac12} J_\pm,
\ee
and the antipode of $\UQ$ is recovered. The definition of the counit $\veps$ goes in the same way. Since we do not use it later, we do not dwell on it.

%COMMENT on the CHOICE of q and the classical limit

%%%%%%%%%%
\subsection{Intertwiners}
%%%%%%%%%%
As we have seen above the algebra of quantum momentum observables  is essentially given by $\UQ$. To construct the Hilbert space of our theory, we consider the representation of $\UQ$ with $q$ real, since we took $q=e^{\hbar \kappa}$. This is convenient since in this case the representation theory of $\UQ$ is essentially the same as the one of $\su(2)$. We refer to the appendix \ref{uq} for the details.

%%%%%%%%%%%%%
\subsubsection{All edges inward}
%%%%%%%%%%%%%

Let us consider a vertex $v$, with $n$ incident edges and attribute an irreducible representation $\cH_{j_i}$ of $\UQ$ to each edge. The total Hilbert space of the vertex $v$ is $\cH_v= \bigotimes_{i=0}^n \cH_{j_i} \equiv \bigotimes_{i=0}^n j_{i} $. As such it is a (reducible) representation of $\UQ$. The action of the generators of $\UQ$ on $\cH_v$ is given by applying $n-1$ times the coproduct, for example $\cop^{(n-1)} J_\pm = (\one\ot ..\ot \one\ot\cop)\circ..\circ (\one\ot \cop)\circ \cop$. Thanks to the coassociativity, there is no issue in grouping the terms.

\smallskip

Note however that the coproduct is not cocommutative, hence we need to order the edges and the Hilbert spaces $\cH_{j_i}$. If we want to change the order, we need to use the deformed permutation map $\psi_{\cR}$. First let us note  $\psi: \UQ\ot \UQ \dr \UQ\ot \UQ$, the usual permutation, then define
\bes\label{deformedPerm}
\psir:  V\ot W &\dr& W\ot V \nn\\
v\ot w &\mapsto &\psir(|v, w\ra)\equiv \psi (\cR |v, w\ra) = \sum \psi (|\cR_{1} v ,\cR_{2} w\ra)  = \sum |\cR_{2} w , \cR_{1}v\ra,
\ees
where we used the notation $\cR=\sum \cR_1 \otimes \cR_2$. The $\cR$-matrix also encodes ``how much'' the coproduct is non-cocommutative, through the relation \cite{chari}
\beq\label{braiding}
(\psi\circ  \cop) X = \cR (\cop X) \cR\mone, \textrm{ with } X \textrm{ a generator of } \UQ.
\eeq
Thanks to this braiding, we recover that the action of the $\UQ$ generators $X$ and the (braided) permutation commute.
\beq\nn
\psir (X (|v, w\ra)) = \psi (\cR X( |v, w\ra))= \psi (\cR (\cop X) |v, w\ra)=  \psi ((\psi \circ \cop X) \cR |v, w\ra)= (\cop X) \psi (\cR |v, w\ra)= X(\psir(|v, w\ra)).
\eeq

\smallskip

Let us consider first the quantization of the trivial Gauss law, namely $\ell_1\ell_2=\one= ((\ell_1\ell_2)^\dagger)\mone$. Since we have two edges, the kinematical Hilbert space is $\cH_v=  j_1 \ot {j_2}$ and a general state will be $i_{j_1j_2}=\sum_{m_1,m_2} i_{m_1,m_2}|j_1m_1,j_2m_2\ra$. The quantum version of $\ell_1\ell_2=\one$ is $\cop \hell =\one\ot \one=\cop \hell^ \dagger \mone $, or in terms of components,
\be
\Delta(K^{\pm1}) = K^{\pm1}\otimes K^{\pm1} =\one ,\qquad \Delta(J_+) = J_\pm\otimes K + K^{-1}\otimes J_\pm =0.
\ee

We want to find the states $i_{j_1j_2}$ that solve these constraints. Looking at the different components of the total quantum angular momentum, we must have
\beq
(\cop J_\pm ) i_{j_1j_2}= 0 , \quad (\cop K^{\pm 1} ) i_{j_1j_2}= i_{j_1j_2}.
\eeq
This is essentially demanding that $i_{j_1j_2}$ is the 2-valent $\UQ$ intertwiner, defined in terms of the Clebsh-Gordon (CG) coefficients,
\be
i^{j_1\,j_2}_{m_1 m_2} = C^{j_1\,j_2\,0}_{m_1 m_2 0} = \delta_{j_1, j_2}\,\delta_{m_1,-m_2}\,\frac{(-1)^{j_1-m_1} q^{\frac{m_1}{2}}}{\sqrt{[2j_1+1]}}.
\ee
As  expected, the quantum Gauss constraint projects $\cH_v$ onto the trivial $\UQ$ representation.

The construction extends naturally for a vertex with $n$ edges. The quantum Gauss constraint reads
\bes\label{ngauss}
\cop^{(n-1)}\hell& =&
%((\one\ot..\ot \one\ot\cop)\circ..\circ \cop) \widehat{\ell} =
\begin{pmatrix} K\otimes \dotsb \otimes K & 0\\ (q^{\frac12}-q^{-\frac12})\left(\sum_{i=1}^n K^{-1} \otimes \dotsb \otimes K^{-1} \otimes J_{+} ^{(i)}\otimes K\otimes \dotsb \otimes K\right)  & K^{-1}\otimes \dotsb \otimes K^{-1} \end{pmatrix} \nn\\
&=& \one\ot\dotsb\ot\one,
\ees
and similarly for $\hell^ \dagger \mone $. We need to find the states $i_{j_1..j_n}$ that solve the constraint, meaning
\beq
(\cop^{(n-1)}J_\pm) i_{j_1..j_n}=0, \quad (\cop^{(n-1)}K^{\pm1}) i_{j_1..j_n}=i_{j_1..j_n}.
\eeq
This is the definition of an invariant vector in $\mathcal{H}_v$, i.e. an intertwiner $j_1\otimes\dotsb\otimes j_n\to \C$. We will focus on the 3-valent case, since this is the situation we will have to deal with later on, and like when $q=1$, generic intertwiners can be split into sums of products of 3-valent intertwiners. Expanding a state $i_{j_1j_2j_3} = \sum_{m_i} i^{j_1\,j_2\,j_3}_{m_1 m_2 m_3} |j_1 m_1, j_2 m_2, j_3 m_3\rangle\in j_1\otimes j_2\otimes j_3$, the Gauss law, $\cop^{(2)}\hell\ i_{j_1 j_2 j_3}=i_{j_1 j_2 j_3}$, produces recursions on the coefficients $i^{j_1\,j_2\,j_3}_{m_1 m_2 m_3}$, solved by CG coefficients,
\be
i_{j_1j_2j_3} = \sum_{m_i} (-1)^{j_3-m_3} q^{-\frac{m_3}{2}}\,C^{j_1\ j_2\ j_3}_{m_1 m_2 -m_3} |j_1 m_1, j_2 m_2, j_3 m_3\rangle.
\ee
This state solves the quantum version of the Gauss constraint $\ell_1\ell_2\ell_3=\one$ and $(\ell^\dagger_3\ell^\dagger_2\ell^\dagger_1)\mone=\one$.

%%%%%%%%%%%%%
\subsubsection{Reverse orientation and dualization}
%%%%%%%%%%%%%

The Gauss constraints we have just discussed corresponds to the case where all the edges are oriented inwards. If however an edge is outgoing, say without loss of generality the edge 1, the classical Gauss constraint is $\tell_1^{-1} \ell_2  \ell_3 =\one$. The quantization of $\tell$ is exactly the same as for $\ell$. $\tell_1^{-1}$ corresponds to dealing with the inverse momentum, hence at the quantum level, we need to use the antipode. It is thus necessary to compose the coproduct with the antipode on the edge 1,
\beq
\tell_1^{-1} \ell_2  \ell_3  \dr \begin{pmatrix} S(K)\otimes K \otimes K & 0\\ (q^{\frac12}-q^{-\frac12})\left(S(J_+)\ot K\ot K + S(K\mone)\ot J_+ \ot K+  S(K\mone)\ot K\mone  \ot J_+  \right)  & S(K^{-1})\otimes K\mone \otimes K^{-1} \end{pmatrix}.
\eeq

However, the antipode of a representation is not a representation. Just like for $q=1$, it is nevertheless a representation on the \emph{dual} Hilbert space. To turn a vector living in the representation $j$ to a vector in the dual representation $j^*$, dualization is realized through the $\UQ$-invariant bilinear form induced by the CG coefficients projecting on the trivial representation. Given two vectors $|u_i\rangle = \sum_m u_{i}^m |j m\rangle$, the bilinear form $B$ is defined as
\be \label{Bform}
B(u_1,u_2) = \sum_{m_1,m_2} g_{m_1\, m_2} \, u_{1}^{m_1}\,u_{2}^{m_2} =  \sum_{m_1,m_2} (-1)^{j-m_1} q^{\f{m_1}{2}} \delta_{m_1, \, -m_2} \, u_{1}^{m_1}\,u_{2}^{m_2}.
%
%   \sqrt{[d_j]}\; C^{j\, j\, 0}_{m_1 m_2 0}\, u_{1}^{m_1}\,u_{2}^{m_2} = \sum_m (-1)^{j+m} q^{-\frac{m}{2}}\,u_{1}^{-m}\ u_{2}^{m}.
\ee
Since  $B$ is neither symmetric nor anti-symmetric\footnote{At $q=1$, it has the property $B(u,v) = (-1)^{2j} B(v,u)$ if $u,v$ are in the representation of spin $j$,
For details on the symmetry properties for $q\neq 1$, we refer to \cite{ours2}.} we can define two types of covector, according to the side we contract $g_{m_1\, m_2}$.
\be\label{DualCoVector}
\langle u^*| = \sum_m (-1)^{j+m}q^{-\frac{m}{2}} u_{-m} \langle jm|, \quad
\langle \overline u| = \sum_m (-1)^{j-m}q^{\frac{m}{2}} u_{-m} \langle jm| \quad \dr \quad B(u_1,u_2) = \langle u_1^*|u_2\rangle = \langle \ov u_2|u_1\rangle.
\ee
In the following, we shall use mostly the notion of duality given by $u^*$.

If $|u\rangle$ transforms with $J_\pm, K$, and denoting $J_\pm |u \rangle\equiv | u_\pm\ra$ and $K|u \ra \equiv |u_K\ra$, then
\be
\la u_\pm^\ast |= \la u^\ast | S(J_\pm),\qquad \text{and}\qquad \la u_K^\ast|= \la u^\ast | S(K),
\ee
i.e. the co-vector $\la u^\ast|$ transforms with the antipode of the generators, $S(J_\pm) = -q^{\pm\frac12}J_\pm$ and $S(K)=K^{-1}$, as requested.

The equation \eqref{DualCoVector} shows that $*$ dualization amounts to changing the components $u_m$ to $(-1)^{j-m}q^{-\frac{m}{2}} u_{-m}$.% (up to a sign $(-1)^{2j}$).
 We conclude that the states in $j_1^*\otimes j_2\otimes j_3$ which satisfy the quantum analogue of Gauss law $\tell_1^{-1} \ell_2 \ell_3=\one$ are proportional to
\be
i_{j_1^* j_2 j_3} = \sum_{m_i}  (-1)^{j_1-m_1} q^{-\frac{m_1}{2}}\ (-1)^{j_3-m_3} q^{-\frac{m_3}{2}}\ C^{j_1\ j_2\ j_3}_{-m_1 m_2 -m_3}\ \langle j_1m_1|\otimes |j_2 m_2\rangle\otimes |j_3 m_3\rangle.
\ee
%\sqrt{\frac{[2][d_{j_1}]}{[2d_{j_1}][d_{j_3}]}}\,
Similarly with the constraint $\ell_1 \tell_2^{-1} \ell_3 = \one$ on $j_1\otimes j_2^*\otimes j_3$,
\be
i_{j_1 j_2^* j_3} = \sum_{m_i}  (-1)^{j_2-m_2} q^{-\frac{m_2}{2}}\ (-1)^{j_3-m_3} q^{-\frac{m_3}{2}}\ C^{j_1\ j_2\ j_3}_{m_1 -m_2 -m_3}\ | j_1 m_1\rangle\otimes \langle j_2 m_2|\otimes |j_3 m_3\rangle,
\ee
and on $j_1\otimes j_2\otimes j_3^*$ with $\ell_1\ell_2\tell_3^{-1}=\one$,
\be \label{123*}
i_{j_1 j_2 j_3^*} = \sum_{m_i}  C^{j_1\ j_2\ j_3}_{m_1 m_2 m_3}\ | j_1 m_1\rangle\otimes |j_2 m_2\rangle \otimes \langle j_3 m_3|.
\ee

Dualization also applies to operators, and this is something we will need.  Assume that $f:V\to W$ is a linear map between two $\UQ$ modules with invariant bilinear forms $B_V, B_W$. Its adjoint $f^\dagger: W \to V$ is defined by
\be
B_V\bigl(v,f^\dagger(w)\bigr) = B_W\bigl(f(v),w\bigr).
\ee
Setting $V=j$, $W=l$, Hilbert spaces carrying the irreducible representations with spin $j,l$, we know that $B_V= \sqrt{[d_j]}\,C^{jj0}_{mn0}, B_W = \sqrt{[d_l]}\,C^{ll0}_{qr0}$ are invariant bilinear forms. Just as for vectors, we can define two types of adjoint, according to the vector duality we choose. In the following, we use the duality induced by $*$, hence the associated duality for operators, noted $\dagger$ is
\be \label{OperatorDualization}
\langle j m|f^\dagger|ln\rangle = (-1)^{j+m} q^{\frac{m}{2}}\,(-1)^{l+n} q^{-\frac{n}{2}}\ \langle l,-n|f|j,-m\rangle.
\ee

\medskip

To sum up, we have obtained the quantization of the Gauss law, for any orientation of the edges, and showed that the $\UQ$ intertwiners are naturally the states satisfying this quantum constraint.

%%%%%%%%%%%%%%%%%
\section{Quantization of vectors and their scalar products} \label{sec:Vectors}
%%%%%%%%%%%%%%%%%

%%%%%%%%%%%%%%%%%
\subsection{Quantization of vectors}
%%%%%%%%%%%%%%%%%

Classically, there is a \emph{left action} of $\SU(2)$ on $\ell\in\SB$ (and $u\in\SU(2)$). For $v\in\SU(2)$, the $\SB$ element $\ell$ is multiplied on the left by $v$, but the result is in $\SL(2,\C)$ rather $\SB$. To identify a well-defined transformation on $\SB$, we have to make use of the Iwasawa decomposition on the product $v\ell$. There exists a unique $v'\in\SU(2)$ and a unique lower triangular matrix $\ell'$ such that
\be \label{SU2onEll}
%\ell\ \mapsto\
v \ell = \ell' v'.
\ee
This gives the transformations
\be
\ell\,\mapsto  v \ell v'^{-1}, \qquad\text{and therefore}\quad \ell \ell^\dagger\, \mapsto  v\, \ell \ell^\dagger \, v\mone.
\ee
This transformation is obviously implemented on the cell decomposition by the Gauss law at the target vertex of the edge carrying $\ell$ and $u$, as explained in the Section \ref{sec:Classical}.

It comes that the product $\ell \ell^\dagger$ transforms in the \emph{adjoint representation} of $\SU(2)$, when considering transformations on the left as in \eqref{SU2onEll}. This implies that the trace $\tr \ell\ell^\dagger$ is $\SU(2)$ invariant, and moreover $\vec{T}\equiv \tr (\ell\ell^\dagger \vec{\sigma})$ is a 3-vector, i.e. it lives in the vector representation. This is equivalent to the representation of spin 1, and the components of $\vec{T}$ in the spherical basis\footnote{We use the standard notations, $\sigma_z = \left(\begin{smallmatrix} 1&0\\0&-1\end{smallmatrix}\right), \sigma_+ = \left(\begin{smallmatrix} 0&1\\0&0\end{smallmatrix}\right), \sigma_- = \left(\begin{smallmatrix} 0&0\\1&0\end{smallmatrix}\right)$.} are
\be \label{TComponents}
T_1 = \tr \ell\ell^{\dagger} \sigma_+,\qquad
\sqrt{2}\,T_0 = -\tr \ell\ell^\dagger \sigma_z,\qquad
T_{-1} = -\tr \ell\ell^\dagger \sigma_-.
\ee

Now let us investigate the quantum version. The key object is the matrix
\be
\ell\ell^\dagger = \begin{pmatrix} \lambda^2 & \lambda \bar{z} \\ \lambda z & \lambda^{-2}+|z|^2\end{pmatrix}.
\ee
The quantization \eqref{QuantumEll} provides $\widehat{\ell}$ and $\widehat{\ell}^{\dagger-1}$. To quantize $\ell^\dagger$, we  use the antipode on $\widehat{\ell}^{\dagger-1}$. This leads to
\be
\begin{aligned}
\widehat{\ell}\ S(\widehat{\ell}^{\dagger-1}) &= \begin{pmatrix} K & 0\\ (q^{\frac12}-q^{-\frac12})J_+ & K^{-1} \end{pmatrix} \begin{pmatrix} K & (q^{\frac12}-q^{-\frac12})q^{-\frac12} J_- \\ 0 & K^{-1} \end{pmatrix}\\
&= \begin{pmatrix} K^2 & (q^{\frac12}-q^{-\frac12})q^{-\frac12} K J_- \\ (q^{\frac12}-q^{-\frac12}) J_+ K & (q^{\frac12}-q^{-\frac12})^2 q^{-\frac12} J_+ J_- + K^{-2} \end{pmatrix}.
\end{aligned}
\ee
We want to investigate whether the quantization of the 3-vector $\vec{T} = \tr\ell\ell^\dagger\vec{\sigma}$ leads to something that it is still a vector (of operators). We form the combinations
\be
t_1 = \frac{q^{\frac12}}{q^{\frac12}-q^{-\frac12}}\,\tr \widehat{\ell}\ S(\widehat{\ell}^{\dagger-1}) \sigma_+,\qquad
\sqrt{[2]}\,t_0 = -\frac{q^{\frac12}}{q^{\frac12}-q^{-\frac12}}\,\tr \widehat{\ell}\ S(\widehat{\ell}^{\dagger-1}) \sigma_z,\qquad
t_{-1} = -\frac{q^{\frac12}}{q^{\frac12}-q^{-\frac12}}\,\tr \widehat{\ell}\ S(\widehat{\ell}^{\dagger-1}) \sigma_-,
\ee
meant to be the quantum version of the components in  \eqref{TComponents}. Using the commutation relations of the generators, a short calculation shows that
\be\label{t comp}
t_1 = KJ_+,\qquad t_0 = -\frac{1}{\sqrt{[2]}}\,(q^{-\frac12}J_+J_- - q^{\frac12}J_-J_+),\qquad t_{-1} = -KJ_-.
\ee
Remarkably, $\bt=(t_{-1},t_0,t_1)$ is recognized as a \emph{vector operator} (see for example \cite{ours2}). It is a vector of operators, whose transformation as a vector (on the index $A=-1,0,1$ in the representation of spin 1) is equivalent to its transformation as a matrix under the adjoint representation\footnote{The adjoint action of a generator $X$ of $\UQ$ on the operator $\cO$ is defined as $X\act\cO = X_{(1)} \cO S(X_{(2)}) $, with $\cop X= X_{(1)}\ot X_{(2)}$. } of $\UQ$,
\be \label{TransfoVectorOperator}
K\triangleright t_A = q^{\frac{A}{2}}\,t_A = K\,t_A\,K^{-1},\qquad \text{and}\qquad J_\pm \triangleright t_A = \sqrt{[1\mp A][1\pm A+1]}\,t_{A\pm1} = J_\pm\,t_A\,S(K) + K^{-1}\,t_A\,S(J_\pm).
\ee
Therefore \emph{the vector operator $\bt$ is the quantization of the classical vector $\vec{T}$}.

From the Wigner-Eckart theorem, it follows that its matrix elements are proportional to CG coefficients,
\be \label{WignerEckartT}
\langle jn|t_A|jm\rangle = N_j\ C^{1\,j\,j}_{A m n},\qquad \text{with $N_j = \sqrt{\frac{[2j][2j+2]}{[2]}}$}.
\ee
When $q\to 1$, the vector operator remains a vector operator: it is formed by the generators of $\SU(2)$, which transform with the $q\to 1$ limit of \eqref{TransfoVectorOperator}.
$$t_A \underset{q\to 1}{\dr} (J_+,-\sqrt{2} J_z, -J_-), \quad J_z\triangleright t_A = [J_z,t_A]= A\,t_A, \quad  J_\pm \triangleright t_A = [J_\pm, t_A] = \sqrt{(1\mp A)(1\pm A+1)}\,t_{A\pm1}.$$
The vector operator $\bt$ is a set of  operators, each of them being represented as a square matrix. We can have more general vector operators with components given by  rectangular matrices. We note $\tau$ such set of operators which still satisfies
\be \label{TransfoVectorOperator bis}
K\triangleright \tau_A = q^{\frac{A}{2}}\,\tau_A = K\,\tau_A\,K^{-1},\qquad \text{and}\qquad J_\pm \triangleright \tau_A = \sqrt{[1\mp A][1\pm A+1]}\,\tau_{A\pm1} = J_\pm\,\tau_A\,S(K) + K^{-1}\,\tau_A\,S(J_\pm),
\ee
with the matrix elements  given by
\be \label{WignerEckartTgeneral}
\langle Jn|\tau_A|jm\rangle = N_{jJ} \ C^{1\,j\,J}_{A m n}.%,\qquad \text{with $N_j = \sqrt{\frac{[2j][2j+2]}{[2]}}$}.
\ee
Due to the triangular constraints from the CG coefficient, we can have $J=j-1,j,j+1$. When $J=j$, we recover $\bt$ with $N_j=N_{jJ} \delta_{jJ}$.

\medskip

At the classical level, another vector-like quantity can be built from $\ell$ and $\ell^\dagger$: $\vec{T}^{\operatorname{op}}\equiv \tr (\ell^\dagger \ell \vec{\sigma})$. Explicitely, its components are given by
\be \label{TopComponents}
T^{\operatorname{op}}_1 = \tr \ell^\dagger\ell\sigma_+,\qquad
\sqrt{2}\,T^{\operatorname{op}}_0 = \tr \ell^\dagger\ell\sigma_z,\qquad
T^{\operatorname{op}}_{-1} = -\tr \ell^\dagger\ell\sigma_-.
\ee
However, while $\ell\ell^\dagger$ transforms as $v\ell\ell^\dagger v^{-1}$, it turns out that $\ell^\dagger \ell$  transforms under the adjoint action of $v^\prime$ instead of $v$,
\be
\ell^\dagger \ell\mapsto v' \ell^\dagger \ell v'^{-1},
\ee
where $v'$ is defined by \eqref{SU2onEll}. As a consequence, the quantized version of $\vec{T}^{\operatorname{op}}$ is not going to transform as a vector under the adjoint action of $\UQ$. In matrix form, we get
\be
S(\widehat{\ell}^{\dagger -1})\,\widehat{\ell} = \begin{pmatrix} K^2 + (q^{\frac12}-q^{-\frac12})^2 q^{-\frac12} J_-J_+ & (q^{\frac12}-q^{-\frac12}) q^{-\frac12} J_- K^{-1} \\ (q^{\frac12}-q^{-\frac12}) q^{-\frac12} J_+ K^{-1} & K^{-2} \end{pmatrix}.
\ee
We proceed just like we did with $\ell \ell^\dagger$, forming the combinations
\be
t^{\operatorname{op}}_1 = \frac{q^{\frac12}}{q^{\frac12}-q^{-\frac12}}\,\tr S(\widehat{\ell}^{\dagger-1})\,\widehat{\ell}\, \sigma_+,\qquad
\sqrt{[2]}\,t^{\operatorname{op}}_0 = \frac{q^{\frac12}}{q^{\frac12}-q^{-\frac12}}\,\tr S(\widehat{\ell}^{\dagger-1})\, \widehat{\ell}\, \sigma_z,\qquad
t^{\operatorname{op}}_{-1} = -\frac{q^{\frac12}}{q^{\frac12}-q^{-\frac12}}\,\tr S(\widehat{\ell}^{\dagger-1})\,\widehat{\ell}\, \sigma_-,
\ee
which are the quantization of the components \eqref{TopComponents} of $\vec{T}^{\operatorname{op}}$. In terms of the generators of $\UQ$, the components read
\be\label{top comp}
t^{\operatorname{op}}_1 = J_+\,K^{-1},\qquad
t^{\operatorname{op}}_0 = \frac1{\sqrt{[2]}}\,(q^{\frac12} J_+ J_- - q^{-\frac12} J_- J_+)
,\qquad
t^{\operatorname{op}}_{-1} = -J_-\,K^{-1}.
\ee
As expected, $\bt^{\operatorname{op}}=(t^{\operatorname{op}}_{-1},t^{\operatorname{op}}_{0},t^{\operatorname{op}}_{1})$ does not transform as a vector under the adjoint action of $\UQ$. It however does transform as a vector under the \textit{coadjoint action}\footnote{The coadjoint action of a generator $X$ of $\UQ$ on the operator $\cO$ is defined as $X\blacktriangleright\cO = S(X_{(1)}) \cO X_{(2)} $, with $\cop X= X_{(1)}\ot X_{(2)}$. } of $\UQ$,
\be \label{EquivarianceTop}
K\blacktriangleright\,t^{\operatorname{op}}_A = q^{-\frac{A}{2}}\,t^{\operatorname{op}}_A = S(K) t^{\operatorname{op}}_A K,\qquad J_\pm \blacktriangleright\, t^{\operatorname{op}}_A = \sqrt{[1\mp A][1\pm A+1]}\,t^{\operatorname{op}}_{A\pm1} = S(J_\pm) t^{\operatorname{op}}_A K + S(K^{-1}) t^{\operatorname{op}}_A J_\pm.
\ee
The reason for this equivariance property can be traced back to the fact that already at the classical level, $\ell^\dagger \ell$ transforms under the coadjoint action of $\SU(2)$ when considering the \emph{right action} of $\SU(2)$ on $\SB$. That is, the right action of $\SU(2)$ is $\ell \longrightarrow \ell \tilde{v}$ for $\tilde{v} \in \SU(2)$, which implies
\be
\ell^\dagger \ell \mapsto \tilde{v}^{-1} \ell^\dagger \ell \tilde{v},
\ee
and this is the coadjoint action.

In the limit $q\to 1$, we recover again that the $\bt^{\operatorname{op}}$ components are proportional to the generators of $\su(2)$, which transforms with the $q\to 1$ limit of \eqref{EquivarianceTop},
\be
t^{\operatorname{op}}  \underset{q\to 1}{\dr}  (J_+,\sqrt{2} J_z, -J_-), \quad J_z\blacktriangleright t^{\operatorname{op}}_A = - [J_z,t^{\operatorname{op}}_A] =-A \,t^{\operatorname{op}}_A , \quad J_\pm \blacktriangleright\, t^{\operatorname{op}}_A = \sqrt{[1\mp A][1\pm A+1]}\,t^{\operatorname{op}}_{A\pm1} = -[J_\pm, t^{\operatorname{op}}_A ]
\ee
%It , i.e. transforms like $, J_\pm \blacktriangleright t^{\operatorname{op}}_A = -[J_\pm, t^{\operatorname{op}}_A]$.
%Notice that for $q=1$, $\bt$ and $\bt^{\operatorname{op}}$ are basically the same vector operators, since $t^{\operatorname{op}}_A = (-1)^{1-A}t_A$.
%\be \label{top=tForQ=1}
%\langle jn|t^{\operatorname{op}}_A|jm\rangle = \sqrt{\frac{2j(2j+2)}{2}}\ C^{1\ \, j\ \ j}_{A -n -m} = \sqrt{\frac{2j(2j+2}{2}}\ (-1)^{1-A}\,C^{1\  j\  j}_{A m n} = (-1)^{1-A}\,\langle jn|t_A|jm\rangle,\qquad \text{for $q=1$}.
%\ee
This is expected since $\vec T$ and $\vec T^{\operatorname{op}}$ coincide up to the factor $(-1)^{1-A}$ already at the classical level when $\kappa\to 0$.
The two vectors $\bt$ and $\bt^{\operatorname{op}}$ can  be related thanks to the antipode, namely  $t^{\operatorname{op}}_A=(-1)^Aq^{-\f{A}{2}} S(t_A)$. When  $q=1$, this simplifies to   $t^{\operatorname{op}}_A = (-1)^{1-A}t_A$, just as  $\vec T$ and $\vec T^{\operatorname{op}}$ when $\kappa\to 0$.

\medskip

It is also worth mentioning that $\bt^{\operatorname{op}}$ is related to the conjugate tensor  $\bar{\bt}$, defined by  in an analogous way as in \eqref{DualCoVector} \cite{biedenharn}
\be \label{conjugate}
\bar{t}_A=(-1)^{1-A} q^{\f{A}{2}} t_{-A}= \left( \tabl{c}{-J_-K \\ \f{1}{\sqrt{2}} (q^{-\f12} J_+ J_- -q^{\f12} J_-J_+) \\ J_+ K}\right).%, \quad m \in \{-1,0,1\}.
\ee
$\bar{\bt}$ transforms under the coadjoint action of $\cU_{q^{-1}}(\su(2))$ as a covector,
\be \label{TransfoConjugate}
K\blacktriangleright_{q^{-1}} \bar{t}_A= q^{\f{A}{2}} \bar{t}_A= \bar{S}(K) \bar{t}_A K, \qquad J_{\pm} \blacktriangleright_{q^{-1}} \bar{t}_A=\sqrt{[1\pm A][1\mp A +1}\, \bar{t}_{A\mp1}= \bar{S}(K) \bar{t}_A J_{\pm} + \bar{S}(J_{\pm}) \bar{t}_A K^{-1}.
\ee
This action comes from the coproduct for $\cU_{q^{-1}}(\su(2))$, $\bar{\Delta}= \psi \circ \Delta = \cR \Delta \cR^{-1}$. $\bar{S}= S^{-1}$ defines the  Hopf algebra antipode operator for $\cU_{q^{-1}}(\su(2))$. And the relationship between $\bar{\bt}$ and $\bt^{\operatorname{op}}$ is simply given by
\be \label{conjugateAndtop}
t_A^{\operatorname{op}}= \bar{t}_{-A} \textrm{ with } q \textrm{ replaced by } q\mone.
\ee
%where $q \rightarrow q^{-1}$ means that we replace $q$ by $q^{-1}$ in $\bar{t}_{-m}$.%, where $t_A$ satisfies the equivariance property \eqref{TransfoConjugate}.
This identification between $\bar{\bt}$ and $\bt^{\operatorname{op}}$ will be important to defined the Wigner matrices entering in the definition of the Hamiltonian constraint in the next section.

\medskip

Writing down the matrix elements of the relations \eqref{EquivarianceTop}, some recursions on the matrix elements of $t^{\operatorname{op}}_A$ are obtained. They lead to a Wigner-Eckart theorem adapted to the coadjoint action \eqref{EquivarianceTop},
\be \label{WignerEckartTop}
\langle jn|t^{\operatorname{op}}_A|jm\rangle = N_j\ C^{1\ \, j\ \ j}_{A -n -m},\qquad \text{with $N_j = \sqrt{\frac{[2j][2j+2]}{[2]}}$}.
\ee
The matrix elements of $\bt^{\operatorname{op}}$ are simply related to the matrix elements of $\bt^\dagger$, the adjoint of $\bt$. Indeed, using \eqref{OperatorDualization},
\be \label{topDual}
\langle jn|t^{\dagger}_A|jm\rangle =N_j (-1)^{n-m} q^{\f{n-m}{2}} C^{1 \; j \; j}_{A\,-n \, -m}= (-1)^A q^{\f{A}{2}}\langle jn|t^{\operatorname{op}}_A|jm\rangle.
\ee

\medskip

The operator $\bt^{\operatorname{op}}$ is a set of  operators, each of them being represented as a square matrix.  One can generalize this to have the components represented by  rectangular matrices. To this aim, let us introduce $\tau^{\operatorname{op}}$ which satisfies
\be \label{EquivarianceTop bis}
K\blacktriangleright\,\tau^{\operatorname{op}}_A = q^{-\frac{A}{2}}\,t^{\operatorname{op}}_A = S(K) \tau^{\operatorname{op}}_A K,\qquad J_\pm \blacktriangleright\, \tau^{\operatorname{op}}_A = \sqrt{[1\mp A][1\pm A+1]}\,\tau^{\operatorname{op}}_{A\pm1} = S(J_\pm) \tau^{\operatorname{op}}_A K + S(K^{-1}) \tau^{\operatorname{op}}_A J_\pm,
\ee
and whose matrix elements are
\be \label{WignerEckartTopgeneral}
\langle Jn|\tau^{\operatorname{op}}_A|jm\rangle = N_{jJ}\ C^{1\ \, J\ \ j}_{A -n -m}.%,\qquad \text{with $N_j = \sqrt{\frac{[2j][2j+2]}{[2]}}$}.
\ee
The triangular constraints of the CG coefficients impose $J=j-1,j,j+1$. When $J=j$, we recover $\bt^{\operatorname{op}}$  with $N_{jJ} \delta_{jJ}=N_j$. Moreover, $\tau^{\operatorname{op}}$ is related to the conjugate of $\tau$ by $\tau^{\operatorname{op}}_m= \bar{\tau}_{-m}$ where we replace $q$ by  $q^{-1}$ in the definition of $\bar{\tau}$.

\medskip

Up to now, we have considered the case of a single Hilbert space $\cH_j$ on which the general vector operator $\tau$ and the operator $\tau^{\operatorname{op}}$ are represented. In the following, we shall need to consider the generalization of such  operators to the tensor product of these Hilbert spaces $\bigotimes_i {j_i}$. In particular we will need to consider the tensor product of vector operators. Due to the non-cocommutativity of the coproduct, the construction of a vector operator that would act on  ${j_2}$ in ${j_1}\ot {j_2}$ is non trivial. The key tool to construct an object that transforms as a vector  is the deformed permutation $\psir$. We know by construction that $\tau\ot \one$ is a vector operator, then we can permute $\tau$ using the deformed permutation.
\be \label{r}
\,^{(1)}\tau = \tau\ot \one, \quad \,^{(2)}\tau\equiv \psir \circ \,^{(1)}\tau \circ  \psir\mone = \cR_{21}(\one\ot \tau)\cR_{21}\mone.
\ee
Thanks to the braided permutation, $\,^{(2)}\tau$ is still a vector operator under the adjoint action of $\UQ$.
%The same will apply for $\tau^{\operatorname{op}}$, we  construct $\,^{(2)}\tau^{\operatorname{op}}\equiv \cR_{21}(\one\ot \tau^{\operatorname{op}})\cR_{21}\mone$ which will be a vector under the coadjoint action of $\UQ$.
The construction is extended to a general tensor product of Hilbert spaces.
\be
\,^{(i)}\tau\equiv \cR_{i\, i-1}\dotsm\cR_{i1}(\one\ot\dotsb\ot\one\ot \tau)\cR_{i1}\mone \dotsm\cR_{i\, i-1}\mone, %\quad \,^{(i)}\tau^{\operatorname{op}}\equiv \cR_{i\, i-1}...\cR_{i1}(\one\ot..\ot\one\ot \tau^{\operatorname{op}})\cR_{i1}\mone  \cR_{i\, i-1}\mone,
\ee
where we used the standard notation $\cR_{12}= \cR_1\ot \cR_2$, and $\cR_{13}= \cR_1\ot \one \ot \cR_3$ and so on. %The use of the deformed permutation makes sure that the nice transformation properties of $\,^{(1)}\tau $ are carried over.
For a recent review of this construction, see \cite{ours2}. Let us recall here the explicit expressions of $\,^{(1)}\bt $ and $\,^{(2)}\bt $ \cite{ours2} that we shall use in the next section. The components of $\,^{(1)}\bt $ are trivially obtained,
\bes
 \,^{(1)}t_+&=& t_+ \otimes \one= KJ_+\otimes \one, \,\,
    \,^{(1)}t_0=t_0 \otimes \one = -\frac{1}{\sqrt{[2]}}\,(q^{-\frac12}J_+J_- - q^{\frac12}J_-J_+)\otimes \one, \,\,
  \,^{(1)}t_-=t_- \otimes \one = -KJ_-\otimes \one. \nn \ees
 The components   $\,^{(2)}t_A=\cR_{21}(\one \otimes t_A)\cR_{21}^{-1}$ are more complicated:
  \bes
 %&&\nn \\
&& \,^{(2)}t_+= % \left(\cR_{21}(\one \otimes t_+)\cR_{21}^{-1}\right)=
 K^2 \otimes K J_+ , \quad
  \,^{(2)}t_0 %= \left(\cR_{21}(\one \otimes t_0)\cR_{21}^{-1}\right)
= -\f{1}{\sqrt{[2]}}\left[\one\otimes(q^{-\f12}J_+J_--q^{\f12}J_-J+) - (q^{\f12}-q^{-\f12})(q^{\f12}+q^{-\f12}) K J_- \otimes KJ_+  \right],\nn \\
&&  \,^{(2)}t_-=% =\left(\cR_{21}(\one \otimes t_-)\cR_{21}^{-1}\right)=
- K^{-2} \otimes KJ_- - q^{-\f12} (q^{\f12} -q^{-\f12})K^{-1}J_- \otimes(q^{-\f12} J_+J_--q^{\f12} J_-J_+) +q^{-\f12}(q^{\f12}-q^{-\f12})^2 (J_-)^2 \otimes K J_+.\label{t2}
\ees

%%%%%%%%%%%%%%%%%
\subsection{Scalar operators} \label{sec:ScalarProd}
%%%%%%%%%%%%%%%%%

The (classical) scalar product provides invariant quantities from vectors. The quantization of scalar products is therefore expected to give rise to scalar operators, i.e. operators which are invariant under $\UQ$. A scalar operator $\cO$ on $\cH_v=\bigotimes_{i=1}^n {j_i}$ must satisfy
\be
J_\pm \act \cO = (\Delta^{(n-1)}J_\pm) \cO(\cop^{(n-1)} K\mone) - q^{\pm\demi}(\cop^{(n-1)} K\mone )\cO (\cop^{(n-1)} J_\pm)= 0, \quad K \act \cO =(\cop^{(n-1)} K) \cO  (\cop^{(n-1)} K\mone)= \cO.
\ee
As such, the  operator $\cO$ commute with the quantum Gauss constraint. To see this, we just need the following lemma.
\begin{lemma}
Let us consider the scalar operator $\cO$ acting on $\bigotimes_{i=1}^n {j_i}$, then  $\cO$ commutes with $\cop^{(n)} K^{\pm1}$ and $\cop^{(n)} J_\pm$.
\end{lemma}
\begin{proof}
Consider the scalar operator $\cO$ living on $\bigotimes_{i=1}^n {j_i}$. By definition we have that
$(\cop^{(n-1)} K) \cO  (\cop^{(n-1)} K\mone)= \cO.$ Multiplying each side by $\cop^{(n-1)} K$ on the right, we see that $\cO$ commutes with $\cop^{(n-1)} K$. The same argument goes for $\cop^{(n-1)} K\mone $. Consider now  $(\Delta^{(n-1)}J_\pm) \cO(\cop^{(n-1)} K\mone) - q^{\pm\demi}(\cop^{(n-1)} K\mone )\cO (\cop^{(n-1)} J_\pm)= 0$ and once again multiply by $\cop^{(n-1)} K$ on the right. We use then that $\cop^{(n-1)} K^{-1}$ commutes with $\cO$ and that $KJ_\pm K^{-1} = q^{\pm\demi} J_\pm$.
\end{proof}
We see therefore that $\cO$ has to commute with each term in the Gauss constraint \eqref{ngauss}, as expected. %Note that the argument will not change if we use instead the coadjoint action.

\medskip

The notion of scalar product for two vectors in the spherical basis is naturally provided by the CG coefficient $C^{1\ 1\ 0}_{m_1  m_2  0}$. Given two vector operators $\,^{(i)}\tau$ and $\,^{(j)}\tau^\prime$ acting on the Hilbert space $\bigotimes_{k=1}^n j_k$, we project their product on its scalar part using the CG coefficients,
\be \label{QuantumScalarProd}
\,^{(i)}\tau\cdot \,^{(j)}\tau^\prime\equiv\sqrt{[3]}C^{1\ 1\ 0}_{m_1  m_2  0}\,^{(i)}{ \tau}_{m}\, \,^{(j)} {\tau^\prime}_{m_2}= \sum_m (-1)^{1-m}q^{\f{m}{2}} \,^{(i)}{ \tau}_{m}\,  \,^{(j)}{\tau^\prime}_{-m} \ee
When the vectors both act on the same Hilbert space $j_i$, the formula reduces quite dramatically. Indeed, following the previous lemma, this quantum scalar product is a scalar operator. Hence it commutes with the generators of $\UQ$ and with the $\cR$-matrix. The many $\cR$-matrices appearing in the definition of $^{(i)}\tau$ and $^{(i)}\tau^\prime$ thus cancel out,
\be
\,^{(i)}\tau\cdot \,^{(i)}\tau^\prime = \one\ot \dotsb\ot \tau\cdot \tau\ot\one\ot\dotsb\ot\one = \sum_m  \one\ot \dotsb\ot\tau^*_m \tau^\prime_m\ot\one\ot\dotsb\ot\one= \sum_m  \one\ot \dotsb\ot \tau_m \bar{\tau^\prime}_m\ot\dotsb\ot\one.
\ee
In particular the norm of $\tau$ is given by
\be
\tau\cdot \tau | j m \rangle= -N_j^2|j m \rangle=- \f{[2j][2j+1]}{[2]}|j m \rangle, %=\bt^{\operatorname{op}}\cdot\bt^{\operatorname{op}} |jm\rangle,
\ee
% A similar construction can applied when considering the scalar product of $\,^{(i)}\tau^{\operatorname{op}}$ and $\,^{(j)}\tau^{\operatorname{op}}$. This will generate an  operator $\,^{(i)}\tau^{\operatorname{op}}\cdot \,^{(j)}\tau^{\operatorname{op}} $ transforming under the coadjoint action of $\UQ$ as a scalar.
and the norm of $\tau^{\operatorname{op}}$,
\be
\tau^{\operatorname{op}}\cdot \tau^{\operatorname{op}} | j m \rangle= -N_j^2 |j m \rangle =- \f{[2j][2j+1]}{[2]}|j m \rangle. %=\bt^{\operatorname{op}}\cdot\bt^{\operatorname{op}} |jm\rangle,
\ee

%%%%%%%%%%%%%%%%%
\subsubsection{Scalar products between two edges incoming at a vertex}
%%%%%%%%%%%%%%%%%

As it has already been emphasized in the Section \ref{sec:Classical}, when the edges 1 and 2 incident on $v$ are oriented inwards, the momentum $\ell_1$ transforms as $v\ell_1 v'^{-1}$ under $\SU(2)$, where $v'$ is defined like in the equation \eqref{SU2onEll}. However, the momentum $\ell_2$ then transforms with the braided rotation $v'$ instead of $v$, i.e. as $v'\ell_2 v''^{-1}$ (where $v''$ is defined through the Iwasawa decomposition $v'\ell_2 = \ell_2' v''$).

This implies that the scalar product $\vec{T}_1^{\operatorname{op}}\cdot \vec{T}_2$ is $\SU(2)$ invariant. This quantity is interesting for geometric reasons \cite{HyperbolicPhaseSpace}. When the Gauss law is satisfied on a 3-valent vertex, for instance $\ell_1 \ell_2 \tilde{\ell}^{-1}_3=\one$, this scalar product corresponds to the cosine of (hyperbolic) angle between the two edges dual to the edges 1 and 2. Indeed, the Gauss law written in terms of $\SU(2)$ invariant quantities translates into the hyperbolic cosine laws on the triangle dual to the ribbon vertex, which enables the identification of geometric observables (like the angle between two edges) in terms of scalar products. Not only the angle between two edges appears in the hyperbolic cosine law, but it is also enters the Hamiltonian constraint, as the latter relates the extrinsic curvature to the intrinsic geometry. This motivates the quantization of $\vec{T}_1^{\operatorname{op}}\cdot \vec{T}_2$, which we thus expect to give rise to a scalar operator.

A direct quantization of  $\vec{T}_1^{\operatorname{op}}\cdot \vec{T}_2$ using the quantum operators \eqref{t comp}, \eqref{top comp} and the definition of the quantum scalar product \eqref{QuantumScalarProd}
cannot lead to a quantum scalar operator. Indeed, $\,^{(1)}\bt^{\operatorname{op}}$ (defined in \eqref{top comp}) transforms as a vector under the \textit{coadjoint action}, whereas  $\,^{(2)}\bt$ transforms as a vector under the \textit{adjoint action}. Hence, there is no sense in contracting their product using the CG coefficients $C^{1 1 0}$ to project on the invariant component.% Hence, we can not expect to get a scalar operator by combining $\,^{(1)}\bt^{\operatorname{op}}$  and $\,^{(2)}\bt$.

Instead, it was shown in  \cite{ours1, ours2} how $\,^{(1)}\bt\cdot \,^{(2)}\bt$  acting on the 3-valent intertwiner $i_{j_1,j_2j_3^*}$ can be interpreted as the quantum version of the hyperbolic cosine law. This comes from the explicit calculation,
\be \label{QuantumAngle}
\bigl(\,^{(1)}\bt\cdot \,^{(2)}\bt\bigr) \, i_{j_1j_2j_3^*} = qN_{j_1} N_{j_2} \sqrt{[2j_1+1][2j_2+1]} (-1)^{j_1+j_2+j_3+1} \begin{Bmatrix} j_1 &j_1 &1\\ j_2 &j_2 &j_3\end{Bmatrix}\ i_{j_1j_2j_3^*}.
\ee
and the explicit value of the quantum 6j-symbol which appear in the above eignevalue. This suggests that a good definition of the quantum version of $\vec{T}_1^{\operatorname{op}}\cdot \vec{T}_2$, involving $\bt$ and $t^{\operatorname{op}}$, should be related to $\,^{(1)}\bt\cdot \,^{(2)}\bt$. As a matter of fact, using the explicit values of $\,^{(1)}\bt$ and $\,^{(2)}\bt $ given in \eqref{t2}, one can arrive at
\be\label{cool}
\,^{(1)}\bt\cdot \,^{(2)}\bt= q \sum_A\, t^{\operatorname{op}}_{-A}\ot t_{A},
\ee
and this is an equality as \textit{operators}. This is a truly remarkable result, because classically, the scalar product expands as
\be
\begin{aligned}
\vec{T}^{\operatorname{op}}_1\cdot \vec{T}_2 &= -2 \left((\tr\ell^\dagger_1\ell_1 \frac{\sigma_z}{\sqrt{2}}) (-\tr\ell_2\ell^\dagger_2 \frac{\sigma_z}{\sqrt{2}}) + (\tr\ell^\dagger_1\ell_1 \sigma_-) (-\tr\ell_2\ell^\dagger_2\sigma_+) + (-\tr\ell^\dagger_1\ell_1 \sigma_+) (\tr\ell_2\ell^\dagger_2\sigma_-)\right)\\
&= -2 \sum_{A=-1,0,1} \bigl(T^{\operatorname{op}}_1\bigr)_{-A}\, \bigl(T_2\bigr)_A,
\end{aligned}
\ee
where the spherical components of $\vec{T}, \vec{T}^{\operatorname{op}}$ are defined in the Equations \eqref{TComponents}, \eqref{TopComponents}. Thus, the Equation \eqref{cool} shows that it is possible to directly promote the components of $\vec{T}, \vec{T}^{\operatorname{op}}$ to their quantum versions \eqref{t comp}, \eqref{top comp}, while keeping the exact same formula as in the classical scalar product,
\be
\vec{T}^{\operatorname{op}}_1\cdot \vec{T}_2\quad \underset{\text{quantization}}{\rightarrow} \quad \sum_{A=-1,0,1} \bigl(t^{\operatorname{op}}_1\bigr)_{-A}\otimes \bigl(t_2\bigr)_A.
\ee
This leads to a scalar operator (by definition of the left hand side of \eqref{cool}), with the nice feature of having a simple and explicit factorization on ${j_1}\otimes {j_2}$. Moreover, this makes the link between the formalism developed in \cite{ours1, ours2} and the quantization of the lattice model of \cite{HyperbolicPhaseSpace}. (It can also be checked directly that the right hand side of \eqref{cool} commutes with the generators of $\UQ$.)

%We shall consider therefore the operator $\sum_A\, t^{\operatorname{op}}_{-A}\ot t_{A}$ as the quantization of $\vec{T}_1^{{\operatorname{op}}}\cdot \vec{T}_2$. It is a scalar operator by construction.

Strikingly,  the $\cR$-matrices entering in the left hand side of \eqref{cool} through the definition of $\,^{(2)}\bt$ are not explicitly visible in the right hand side. Recalling from \cite{HyperbolicPhaseSpace} that $\vec T^{\operatorname{op}}_1= h_1 \act \vec T_1$ where $h_1$ is a rotation entering the Cartan decomposition of $\ell_1$, and quantizing $\vec T_1$ as the vector operator $\,^{(1)}\bt$, it seems likely that the quantization of $h_1$ is somehow related to the $\cR$-matrix. While this should be further investigated, this is not necessary on a first approach to the quantization of neither the Gauss law, nor the Hamiltonian, and will be studied elsewhere.

%%%%%%%%%%%%%%%%%
\subsubsection{Scalar products between two edges outgoing at a vertex}
%%%%%%%%%%%%%%%%%

We now deal with the case where the edges 1 and 2 are outward at $v$. The Gauss law becomes $\tell^{-1}_1\tell^{-1}_2 \dotsm =\one$ and therefore it is necessary to introduce the vectors based on $\tell$,
\be
\vec{\widetilde{T}} = \tr \tell \tell^\dagger \vec{\sigma},\qquad \vec{\widetilde{T}}^{\operatorname{op}} = \tr \tell^\dagger \tell \vec{\sigma}.
\ee
It is easy to check that $\vec{\widetilde{T}}_1\cdot \vec{\widetilde{T}}^{\operatorname{op}}_2$ is $\SU(2)$ invariant. Following the lesson of the previous section, a straightforward candidate to its quantization is $\sum_A\, \left(t_{-A}\ot t^{\operatorname{op}}_{A}\right)$. On the other hand we know that a well-defined scalar operator on $\Inv(j_1^*\otimes j_2^*\otimes \dotsb)$ is obtained upon dualizing the quantum scalar product \eqref{cool} with respect to 1 and 2. %When studying the Hamiltonian constraint in the next section, we will also have to deal with intertwiners of the form $i_{j_1^\ast j_2^* j_3^*} \in \textrm{Inv}(j_1^*\otimes j_2^* \otimes j_3^*)$. An intertwiner $i_{j_1^\ast j_2^* j_3^*}$ corresponds to a quantized 3-leg vertex where all the edges are outward.  We need then to define a scalar operator well-defined on $\textrm{Inv}(j_1^*\otimes j_2^* \otimes j_3^*)$. This can be done starting from $ \sum_A \, t^{\operatorname{op}}_{-A}\ot  t_{A}$ given in \eqref{cool} which is well-defined on $\textrm{Inv}(j_1\otimes j_2 \otimes j_3^*)$ that we dualize by taking the adjoint operators of $\bt^{\operatorname{op}}$ and $\bt$.
Applying the definition of the adjoint operator \eqref{OperatorDualization} to $t^{\operatorname{op}}$ and using \eqref{topDual} for the definition of $t^\dagger$, we can actually match the two proposals, since
\be \label{toptAdjoint}
\sum_A \, \left(t^{\operatorname{op} \dagger}_{-A}\ot  t^\dagger_{A} \right)= \sum_A\, \left(t_{-A}\ot t^{\operatorname{op}} _{A}\right).
\ee

%This result is again consistent with what we expect from the classical theory. Indeed, at the classical level, when the edges 1,2 are outward, %Since $\tell$ is quantized like $\ell$, the quantization is straightforward,
We will henceforth use the quantization map
\be
\vec{\widetilde{T}}_1\cdot \vec{\widetilde{T}}^{\operatorname{op}}_2 \qquad \rightarrow \qquad \sum_{A=-1,0,+1} (t_1)_{-A}\otimes (t_2^{\operatorname{op}})_{A}
\ee
This invariant operator is diagonal on the spaces $\Inv(j_1^*\ot j_2^* \ot j_3^*)$. Let us find its eigenvalues,
\begin{multline}
 \sum_A\, \left(t_{-A}\ot t^{\operatorname{op}} _{A}\right)\ i_{j_1^* j_2^* j_3^*} \\
= N_{j_1} N_{j_2} \sum C^{\,1\ j_1\ j_1}_{-A n_1 m_1} C^{1\,j_2\ j_2}_{A -m_2 -n_2} (-1)^{j_1-m_1+j_2-m_2} q^{-\frac{m_1+m_2}{2}} C^{\ \,j_1\ \,j_2\ \,j_3}_{-m_1 -m_2 m_3} \langle j_1 n_1|\otimes \langle j_2 n_2| \otimes \langle j_3 m_3|
\end{multline}
We focus on the sum over $A,m_1,m_2$,
\be
\begin{aligned}
&\sum_{A,m_1,m_2} C^{\,1\ j_1\ j_1}_{-A n_1 m_1} C^{1\,j_2\ j_2}_{A -m_2 -n_2} (-1)^{j_1-m_1+j_2-m_2} q^{-\frac{m_1+m_2}{2}} C^{\ \,j_1\ \,j_2\ \,j_3}_{-m_1 -m_2 m_3} \\
&= -\sum_{A,m_1,m_2} C^{\,j_1\ 1\ \,j_1}_{-n_1 A -m_1} (-1)^{1-A}q^{\frac{A}{2}} C^{\,1\ \,j_2\ \,j_2}_{-A -n_2 -m_2} (-1)^{j_1+j_2+m_3} q^{\frac{m_3}{2}} C^{\ \,j_1\ \,j_2\ \,j_3}_{-m_1 -m_2 m_3}\\
%&= (-1)^{j_1+j_2+1+m_3} q^{\frac{m_3}{2}} \begin{Bmatrix} j_1 &1 &j_1\\ j_2 &j_3 &j_2\end{Bmatrix} (-1)^{j_1+j_2+j_3} \sqrt{[2j_1+1][2j_2+1]} C^{\ \,j_1\ \,j_2\ \,j_3}_{-n_1 -n_2 m_3}\\
&= \sqrt{[2j_1+1][2j_2+1]}(-1)^{j_1+j_2+j_3+1} \begin{Bmatrix} j_1 &1 &j_1\\ j_2 &j_3 &j_2\end{Bmatrix} (-1)^{j_1-n_1+j_2-n_2} q^{-\frac{n_1+n_2}{2}}  C^{\ \,j_1\ \,j_2\ \,j_3}_{-n_1 -n_2 m_3}.
\end{aligned}
\ee
This leads to
\be \label{t1*t2*}
 \sum_A\, \left(t_{-A}\ot t^{\operatorname{op}} _{A}\right)\ i_{j_1^* j_2^* j_3^*} = N_{j_1} N_{j_2} \sqrt{[2j_1+1][2j_2+1]}(-1)^{j_1+j_2+j_3+1} \begin{Bmatrix} j_1 &j_1 &1\\ j_2 &j_2 &j_3\end{Bmatrix}\ i_{j_1^* j_2^* j_3^*}
\ee
This is remarkably the same result as \eqref{QuantumAngle}, when the two edges 1,2 are incoming at the vertex.

%%%%%%%%%%%%%%%%%
\subsubsection{Scalar products between one incoming and one outgoing edge}
%%%%%%%%%%%%%%%%%

It is also possible to build invariant operators when only one edge orientation is flipped at the vertex. For instance when the edge 2 is outgoing while the edge 1 is oriented inward, the quantized vertex is given by an intertwiner of the form $i_{j_1 j_2^* j_3^*} \in \textrm{Inv}(j_1\otimes j_2^* \otimes j_3^*)$. The invariant operator acting in $\textrm{Inv}(j_1\otimes j_2^* \otimes j_3^*)$ can once again be built from $\sum_A \left( t^{\operatorname{op}}_{-A}\ot  t_{A}\right)$ and by taking the adjoint of $\bt$ given by \eqref{topDual}   to dualize the operator acting on the second leg of the intertwiner. We obtain that
\be
\sum_A \left( t^{\operatorname{op}}_{-A}\ot  t_{A}^\dagger\right) \,i_{j_1 j_2^* j_3^*} = \sum_A (-1)^{-A} q^{-\f{A}{2}} \left( t^{\operatorname{op}}_{-A}\ot  t^{\operatorname{op}}_{A}\right)\, i_{j_1 j_2^* j_3^*} = -\sum_{A,B} \sqrt{[3]} C^{1\,1\,0}_{B A 0}\ (t^{\operatorname{op}}_{A}\otimes t^{\operatorname{op}}_{B})\, i_{j_1 j_2^* j_3^*} ,
\ee
where $\sqrt{[3]} C^{1\,1\,0}_{B A 0} = \delta_{B,-A} (-1)^{1-B} q^{\frac{B}{2}}$. At the classical level, the invariant vector is given by $\vec{T}^{\operatorname{op}}_1\cdot \vec{\widetilde{T}}^{\operatorname{op}}_2$. Therefore, we find that the quantization  map that is required is more subtle,
\be
\vec{T}^{\operatorname{op}}_1\cdot \vec{\widetilde{T}}^{\operatorname{op}}_2 \qquad \rightarrow \qquad -\sum_{A,B} \sqrt{[3]} C^{1\,1\,0}_{B A 0}\ (t^{\operatorname{op}}_{A}\otimes t^{\operatorname{op}}_{B}).
\ee
%We recall that $\sqrt{[3]} C^{1\,1\,0}_{B A 0} = \delta_{B,-A} (-1)^{1-B} q^{\frac{B}{2}}$. Since in the limit $q=1$ we know the operator must go to $\sum_B (t_1)_{-B} \otimes (t_2)_B$, and from the Equation \eqref{top=tForQ=1} that $t^{\operatorname{op}}_{B|q=1} = (-1)^{1-B} t_{B|q=1}$, the factor $(-1)^{1-B}$ in $C^{1\,1\,0}_{B A 0}$ is understandable. As for the factor $q^{\frac{B}{2}}$, its limit as $q\to 1$ is trivial. Although this requires further investigation, the reason for the CG coefficient $C^{1\,1\,0}_{B A 0}$ is certainly related to the fact that the operators on the edges 1,2 are both $t^{\operatorname{op}}$, and this somehow requires a CG coefficient to project their tensor product on the trivial representation. The same is true if the edge 1 is outgoing and the edge 2 incoming.
The last case is when the edge 1 is outgoing and the edge 2 incoming. The corresponding  invariant operator acting in $\textrm{Inv}(j_1^*\otimes j_2 \otimes j_3^*)$ is then defined by
\be
\sum_A \left( t^{op \dagger}_{-A}\ot  t_{A}\right) \,i_{j_1 j_2^* j_3^*} = \sum_A (-1)^{A} q^{\f{A}{2}} \left( t_{-A}\ot  t_{A}\right)\, i_{j_1 j_2^* j_3^*} = -\sum_{A,B} \sqrt{[3]} C^{1\,1\,0}_{B A 0}\ (t_{A}\otimes t_{B})\,i_{j_1 j_2^* j_3^*},
\ee
Therefore, we recover the same quantization map as in the previous case, since the classical invariant in this case is given by $\vec{\widetilde{T}}_1\cdot \vec{T}_2 $. That is,
\be
\vec{\widetilde{T}}_1\cdot \vec{T}_2 \qquad \rightarrow \qquad \sum_{A,B=-1,0,+1}- \sqrt{[3]} C^{1\,1\,0}_{B A 0}\ (t_{A}\otimes t_{B}).
\ee
Note that in the quantization scheme proposed above, we have focused on operators acting on the first two legs of the intertwiner. Extending the construction to operators also acting  on the third leg is more involved and a better understanding of the use of the $\cR$-matrix together with our quantization scheme is required. Fortunately, we do not need this for the next steps, hence we postpone this issue for ulterior investigations.

%%%%%%%%%%%%%%%%%%%%%%%
\section{Quantization of the holonomies}
%%%%%%%%%%%%%%%%%%%%%%%

In order to study the Hamiltonian constraint (or ``flatness'' constraint $\mathcal{C}_f$), we need to quantize the holonomies $u, \tu$. To restrict the present article to a reasonable size, we will only give the ingredients required to quantize the Hamiltonian on a particular face of degree three (done in the Section \ref{sec:Hamiltonian}), and we postpone a complete study to further reports.

Let us first recall some results about the Wigner matrices for the quantum group $\SU_q(2)$ in the vector representation. The classical matrix element $R^1_{AB}(u)$ of $u\in \SU(2)$ is quantized as the operator given by the $\SU_q(2)$ Wigner matrix element $\,^q\cD^1_{AB}(u)$, $u\in\SU_q(2)$, in the vector representation of $\SU_q(2)$. %which is an element of  $\SU_q(2)$.
To determine the action of this operator on the relevant space Hilbert space,  we use the  \emph{Wigner product law} \cite{biedenharn}. In the following, unless specified otherwise, we have $u\in\SU_q(2)$.
\bes \label{WignerProductLaw}
\,^q\cD^1_{BA}(u) \,^q\cD^{j_1}_{n_1 m_1}(u) &= &\,^q\cD^1_{BA}(u) \la j_1 n_1 | u | j_1 m_1\ra  =\sum_{\substack{J_1 = j_1-1,j_1,j_1+1 \\ -J_1 \leq M_1,N_1,\leq J_1}} C^{1 j_1 J_1}_{B n_1 N_1} \,^q\cD^{J_1}_{N_1 M_1}(u)C^{1 j_1 J_1}_{A m_1 M_1} \nn\\
&=&
\sum_{\substack{J_1 = j_1-1,j_1,j_1+1 \\ -J_1 \leq M_1,N_1,\leq J_1}} C^{1 j_1 J_1}_{B n_1 N_1} \la J_1 N_1 | u | J_1 M_1\ra C^{1 j_1 J_1}_{A m_1 M_1}.
\ees
As such the Wigner matrix element is a map $ j_1^*\otimes j_1 \to \bigoplus_{J_1=j_1,j_1\pm1} J_1^*\otimes J_1$,
\be \label{HolonomyAction}
\,^{q}\cD^{(1)}_{BA}(u)\ \langle j_1 n_1| \otimes |j_1 m_1\rangle = \sum_{J_1, M_1, N_1}  C^{1 j_1 J_1}_{B n_1 N_1} C^{1  j_1 J_1}_{A m_1 M_1} \langle J_1 N_1|\otimes | J_1 M_1\rangle.
\ee

In \eqref{TransfoVectorOperator}, we have defined tensor operators as objects ``transforming well'' under $\UQ$. This property is naturally extended\footnote{This transformation law is the $q$-analogue of the equivariance transformation property for a $\SU(2)$ tensor operator $\bt^j$ with components $\bt^j_m$ where the action by $g \in \SU(2)$ obeys $U(g) \bt^j_m U^{-1}(g)=\sum_{m^\prime} T^j_{m^\prime} R^j_{m^\prime m} (g)$ with $R^j(g)$ the $j$ representation of $\SU(2)$.} to $\SU_q(2)$ \cite{biedenharn}, in which case we use  the Wigner matrix elements $\,^{q}\cD^j_{AB}(u)$. Given a  vector operator $\bt$, we can construct the new vector operator  $\bt^\prime$ using the Wigner matrix element
\be\label{TransfoT}
t_A \longrightarrow t_A^\prime =\sum_B t_B \,^{q}\cD^{1}_{BA}(u).
\ee
Writing the Wigner matrix elements as $\,^{q}\cD^{1}_{BA}(u)\equiv \sum_{J_1} D_B^* \ot D_A: j_1^* \otimes j_1 \longrightarrow \bigoplus_{J_1=j_1,j_1\pm1} J_1^*\otimes J_1$, and looking at the left action of $\SU_q(2)$ on this quantum matrix, it can be shown \cite{biedenharn}, using the group multiplication law\footnote{The group multiplication law is given by, $\,^{q}\cD(u) \,^{q}\cD(\tilde{u})=\,^{q}\cD(u\tilde{u}) \; \forall \, u, \, \tilde{u} \in \SU_q(2)$}, that the columns of $\,^{q}\cD^{1}(u)$, $D_A$, transform according to \eqref{TransfoT}. Therefore, $D_A$ is proportional to a vector operator $\tau_A$ (the vector whose components are rectangular matrices, \eqref{WignerEckartTgeneral}), and consequently, $D_B^*$ is proportional to the conjugate vector operator $\bar{\tau}_B$ (the bar operation is defined in \eqref{conjugate}) in order to have \eqref{TransfoT}. The action of $\,^{q}\cD^{1}_{BA}(u):\, j_1^*\otimes j_1 \to \bigoplus_{J_1=j_1,j_1\pm1} J_1^*\otimes J_1$ is thus given, up to proportionality factor $F(j_1, J_1)$, by
\bes  \nn
\,^{q}\cD^{1}_{BA}(u) \langle j_1 n_1| \otimes |j_1 m_1\rangle & =& \sum_{J_1, N_1, M_1}  F(j_1, J_1) \, \la j_1 n_1 |\bar{\tau}_B| J_1 N_1\ra \la J_1 M_1 | \tau_A | j_1 m_1 \ra \la J_1 N_1 | \otimes |J_1 M_1 \ra  \\
&=& \sum_{\substack{J_1 = j_1-1,j_1,j_1+1 \\ -J_1 \leq M_1,N_1,\leq J_1}} F(j_1, J_1) \, N_{j_1 J_1} N_{J_1j_1}\,  (-1)^{1-B} q^{\f{B}{2}} \, C^{1 J_1 j_1}_{-B N_1 n_1} C^{1 j_1 J_1}_{A m_1 M_1} \langle J_1 N_1|\otimes | J_1 M_1\rangle .\label{tutu}
\ees
where we used that $\bar{\tau}_B=(-1)^{1-B}q^{\f{B}{2}} \tau_{-B}$ and that the tensor operator $\tau$ satisfies the Wigner-Eckart theorem \eqref{WignerEckartTgeneral}. The proportionality coefficient $F(j_1, J_1)$ can be determined by comparing the action of $\,^qD^1$ coming from the Wigner product law \eqref{HolonomyAction} with \eqref{tutu}.
Since $(-1)^{1-B} q^{\f{B}{2}} C^{1 J_1 j_1}_{-B N_1 n_1}= (-1)^{J_1-j_1+1} \sqrt{\f{[2j_1+1]}{[2J_1+1]}} C^{1 j_1 J_1}_{B n_1 N_1}$, we get that
\be\label{ProporCoeff}
F(j_1, J_1)= (-1)^{j_1-J_1+1}\sqrt{\f{[2J_1+1]}{[2j_1+1]}} N^{-1}_{j_1 J_1}N^{-1}_{J_1 j_1} .
\ee

\medskip

Let us now consider an edge, say with the label 1, on our cell decomposition. From the equality $\ell_1 u_1 = \tu_1 \tell_1$, it is found that $u_1$ transforms $\vec{\widetilde{T}}^{\operatorname{op}}_1$ into $\vec{T}^{\operatorname{op}}_1$,
\be\label{cholo}
R(u_1) \vec{\widetilde{T}}^{\operatorname{op}}_1 = \vec{T}^{\operatorname{op}}_1,\ \text{or in components}\ \sum_B R^{1}_{AB}(u_1) \bigl(\widetilde{T}^{\operatorname{op}}_1\bigr)_{-B} = \bigl(T^{\operatorname{op}}_1\bigr)_{-A},
\ee
where $R(u_1)$ is the rotation (in the vector representation). It is crucial that this property still holds upon quantizing $R(u_1)$, to make sure that the quantum holonomy transforms the quantum vector at the source vertex of the edge into the quantum vector at the target vertex of the edge.

We thus need to identify the Wigner matrix relating $\bt^{\operatorname{op}}$ operators. Let us recall that $\bt^{\operatorname{op}}$ is related to $\bar{\bt}$ by \eqref{conjugateAndtop}. Therefore we first extend \eqref{TransfoT} to the $\bar{\bt}$ case, following \cite{biedenharn},
\be \label{equivarianceConjugate}
\bar{t}_A \longrightarrow \bar{t}_A^\prime =\sum_B \bar{t}_B \,^{q}\cD^{\star1}_{AB}(u),%(u) \quad \textrm{ with } u \in \SU_q(2),
\ee
where $\,^{q}\cD^{\star1}(u)$ is the conjugate rotation matrix defined by
\be \label{conjugateMatrix}
\,^{q}\cD^{\star(1)}_{AB}(u) %\left(\,^{q}\cD^1_{B A}(u) \right)^\star
\equiv
(-1)^{A-B} q^{\f12(A-B)} \,^{q}\cD^1_{-B-A}(u).
\ee
We want to consider $\,^{q}\cD^{\star(1)}_{AB}(u)$ as a map
\be
\,^q\cD^{\star1}_{AB}(u)\equiv \sum_{J_1} D_A^\star\ot D_B^{*\star}: j_1\otimes j_1^* \to \bigoplus_{J_1=j_1,j_1\pm1} J_1\otimes J_1^*.
\ee
Comparing the definition of $\,^q\cD^{\star1}(u)$ with the definition of the conjugate tensor operator \eqref{conjugate}, it appears that $D_A^\star$ is the conjugate tensor operator under $\SU_q(2)$. More precisely, the conjugate matrix given in \eqref{conjugateMatrix} can be realized as
\be
\,^q\cD^{\star1}_{AB}(u)=\sum_{J_1} D_A^\star\ot D_B^{*\star}= \sum_{J_1} (-1)^{A-B} q^{\f12(A-B)} \, D_{-A} \otimes D_{-B}^*= \sum_{J_1} \bar{D}_A \otimes ((-1)^{1-B} q^{-\f{B}{2}}D_{-B}^*).
\ee
We see that  $\,^q\cD^{\star1}_{AB}(u)$ is proportional to $\sum_{J_1} \bar{\tau}_A\otimes ((-1)^{1-B} q^{-\f{B}{2}} \bar{\tau}_{-B})$,  where the proportionality coefficient is given by \eqref{ProporCoeff}. Proceeding to the change $q\dr q\mone$, we get the transformation law for $t^{\operatorname{op}}$.%\eqref{equivarianceConjugate} implies for $\bt^{\operatorname{op}}$ that
\be
t^{\operatorname{op}}_{-A} \longrightarrow t^{\operatorname{op} \prime}_{-A}=\sum_B t^{\operatorname{op}}_{-B}\; \,^{q^{-1}}\cD^{\star1}_{AB} (u).%, \quad u \in \SU_q(2).
\ee
Consequently, $\,^{q^{-1}}\cD^{\star 1}(u)_{AB}$ can be written as $\,^{q^{-1}}\cD^{\star 1}(u)_{AB}=\sum_{J_1} D^{\operatorname{op}}_A \ot D^{\operatorname{op}*}_B$ with $D^{\operatorname{op}}_A$ proportional to $\tau^{\operatorname{op}}_{-A}$ and $D^{\operatorname{op}*}_B$ proportional to $(-1)^{1-B}q^{\f{B}{2}}\tau^{\operatorname{op}}_{B}$. Explicitly, we get,
\be \label{QuantumHol}
\,^{q^{-1}}\cD^{*1}_{AB}(u)=\sum_{J_1} D_A^{\operatorname{op}}\otimes D_B^{\operatorname{op}*}\ \text{with } \left\{
\begin{aligned}
&  \langle J_1 M_1|D^{\operatorname{op}}_A|j_1 m_1\rangle = (-1)^{J_1-j_1+1}\sqrt{\frac{[2J_1+1]}{[2j_1+1]}}\,C^{1 J_1 j_1}_{-A -M_1 -m_1}\\
&  \langle j_1 n_1|D_B^{\operatorname{op}*}|J_1 N_1\rangle = (-1)^{1-B} q^{\frac{B}{2}} C^{1 j_1 J_1}_{B -n_1 -N_1}
\end{aligned}\right.
\ee
where we have distributed the proportionality coefficient  \eqref{ProporCoeff} conveniently for the coming calculations.

Let us now check explicitly that the constraint \eqref{cholo} is realized at the quantum level. It comes
\bes
&&\sum_B t^{\operatorname{op}}_{-B} \; \,^{q^{-1}}\cD^{*1}_{AB}(u) |j_1 m_1\ra \otimes \la j_1 l_1|= \sum_{\substack{B, J_1\\ M_1, N_1, n_1}} \la j_1 l_1 |t^{\operatorname{op}}_{-B}| j_1 n_1\ra \la J_1 M_1 | D^{\operatorname{op}}_A|j_1 m_1 \ra \la j_1 n_1 | D^{\operatorname{op}*}_B | J_1 N_1 \ra \,|J_1 M_1\ra \otimes \la J_1 N_1 | \nn \\
&& \qquad \qquad = \sum_{B, J_1, M_1, N_1, n_1} N_{j_1} C^{1 j_1 j_1}_{-B -l_1 -n_1} (-1)^{J_1-j_1+1} \sqrt{\f{[2J_1 +1]}{[2j_1+1]}} C^{1 J_1 j_1}_{-A -M_1 -m_1} (-1)^{1-B}q^{\f{B}{2}} C^{1 j_1 J_1}_{B -n_1 -N_1} |J_1 M_1\ra \otimes \la J_1 N_1 | \nn \\
&& \qquad \qquad=  \sum_{B, J_1, M_1, N_1, n_1} N_{j_1} C^{1 j_1 j_1}_{B -n_1 -l_1} (-1)^{J_1-j_1}\sqrt{\f{[2J_1 +1]}{[2j_1+1]}} C^{1 J_1 j_1}_{-A -M_1 -m_1}  C^{1 j_1 J_1}_{B -n_1 -N_1}   |J_1 M_1\ra \otimes \la J_1 N_1 |\nn \\
&& \qquad \qquad =  \sum_{ J_1, M_1, N_1} N_{j_1} (-1)^{J_1-j_1}\sqrt{\f{[2J_1 +1]}{[2j_1+1]} }C^{1 J_1 j_1}_{-A -M_1 -m_1} \delta_{J_1 j_1} \delta_{l_1 N_1}   |J_1 N_1\ra \otimes \la J_1 N_1 | \nn \\
&&\qquad \qquad = t^{\operatorname{op}}_{-A} |j_1m_1\ra \otimes \la j_1 l_1|
\ees
which is the quantum version of \eqref{cholo}.

%%%%%%%%%%%%%%%%%%%%%%%
\section{Quantization of the flatness constraint} \label{sec:Hamiltonian}
%%%%%%%%%%%%%%%%%%%%%%%

The flatness constraint is supported on faces of the cell decomposition. In this section we will restrict attention to a face of degree 3. Although the generic case probably does not require any major changes, the generalization is not totally straightforward as it is expected to involve coefficients of $\UQ$ recoupling theory with many representations.

%%%%%%%%%%%%%%%%%%%%%%%
\subsection{The Hamiltonian at the classical level}
%%%%%%%%%%%%%%%%%%%%%%%

The situation we are going to analyze is depicted in the Figure \ref{fig:3ValentFace}. The three edges supporting the flatness constraint are labeled 1,2,6 and the vertices are labeled by the three edges incident to each of them, i.e. $v_{123}, v_{264}, v_{156}$. The classical constraint reads
\be \label{3ValentConstraint}
\tu_2^{-1}\,u_1\,u_6 = \one
\ee

\begin{figure}
\includegraphics[scale=.5]{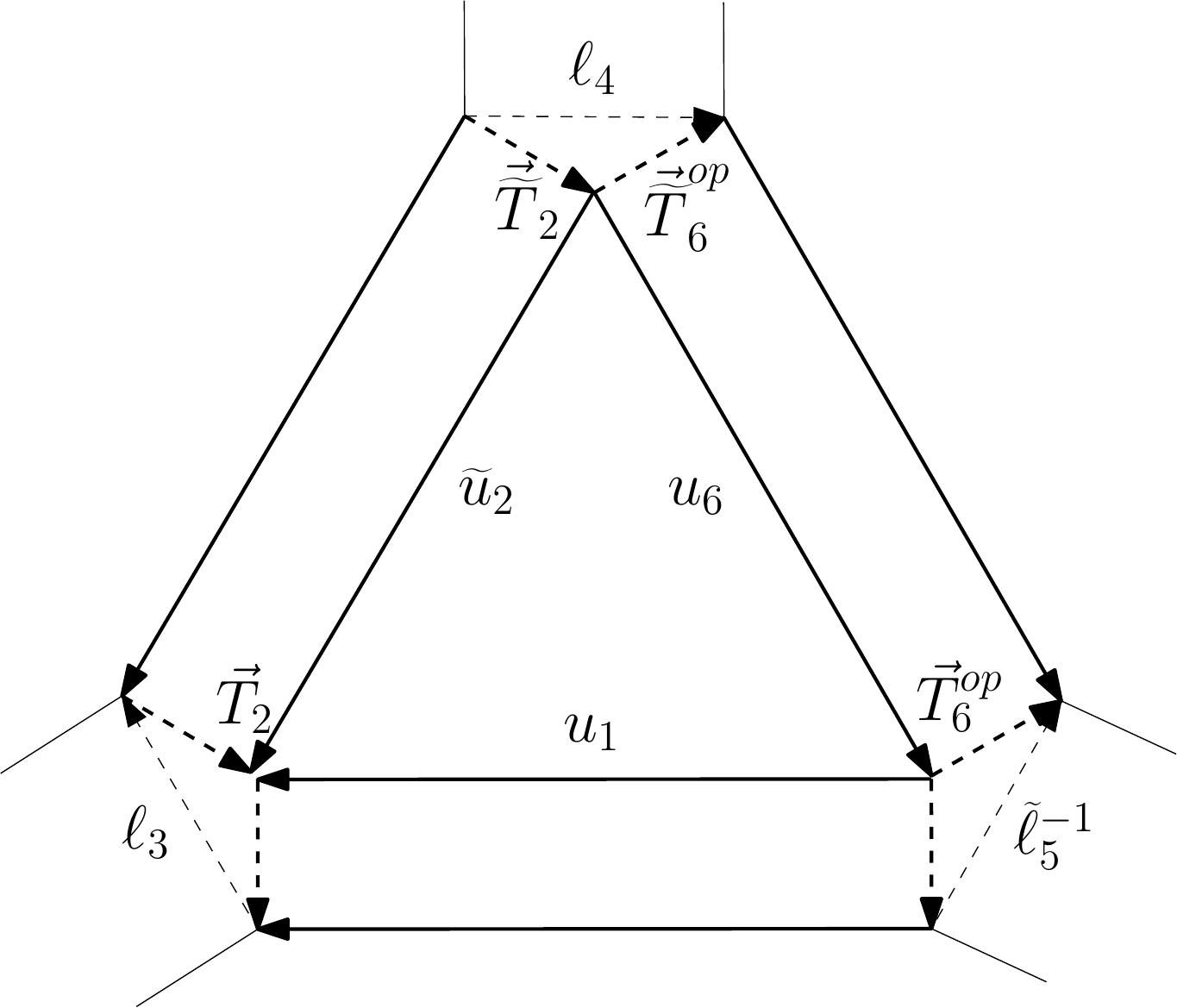}
\caption{\label{fig:3ValentFace} This is a face of degree 3 which is the support of a flatness constraint.}
\end{figure}

We recall that the scalar product $\vec{\widetilde{T}}_2\cdot \vec{\widetilde{T}}^{\operatorname{op}}_6$ is invariant and represents the cosine of the angle between the edges dual to 2 and 6 ($\vec{\widetilde{T}}$ and $\vec{\widetilde{T}}^{\operatorname{op}}$ are defined just like $\vec{T}$ and $\vec{T}^{\operatorname{op}}$ but through $\tell$ instead of $\ell$). A consequence of the constraint is that if this observable is evaluated \emph{after} parallel transport of, say, $\vec{\widetilde{T}}^{\operatorname{op}}_6$ around the face, it should not be affected. That is
\be
\vec{\widetilde{T}}_2\cdot R(\tu_2^{-1}\,u_1\,u_6) \vec{\widetilde{T}}^{\operatorname{op}}_6 - \vec{\widetilde{T}}_2\cdot \vec{\widetilde{T}}^{\operatorname{op}}_6 = 0.
\ee
This equation is labeled by a pair of edges, 2 and 6, or equivalently a pair vertex-face. We propose to define the Hamiltonian constraint this way,
\be \label{Hamiltonian}
H_{26} = \vec{\widetilde{T}}_2\cdot R(\tu_2^{-1}\,u_1\,u_6) \vec{\widetilde{T}}^{\operatorname{op}}_6 - \vec{\widetilde{T}}_2\cdot \vec{\widetilde{T}}^{\operatorname{op}}_6.
\ee
In general, whether there are some $\vec{T}, \vec{T}^{\operatorname{op}}, \vec{\widetilde{T}}$ or $\vec{\widetilde{T}}^{\operatorname{op}}$ depend on the orientations of the two edges at the vertex. This is a reasonable choice of Hamiltonian for the following reasons.
\begin{itemize}
\item The constraint \eqref{3ValentConstraint} has three real degrees of freedom, while the constraint \eqref{Hamiltonian} has only one, but there are three such constraints, one for each vertex around the face. So we get the correct number of degrees of freedom.
\item What we have actually done is rewrite the constraint \eqref{3ValentConstraint} in the vector representation of $\SU(2)$ and in a specific basis formed by the dynamical variables $(\vec{T})$. Here a possible basis is formed by $\vec{\widetilde{T}}_2, \vec{\widetilde{T}}^{\operatorname{op}}_6$ and a third vector which transforms the same way, say $R(\tu^{-1}_2) \vec{T}^{\operatorname{op}}_1$. Generally those are three independent vectors which therefore probe the three degrees of freedom of the flatness constraint\footnote{For a face of degree greater than 3, there are more than 3 such Hamiltonian constraints. However, it can be shown that only three of them are independent using the same method as in the non-deformed case in \cite{3DHamiltonian}.}. Using parallel transport (similar to the manipulation in the next paragraph), each of the three constraints  can be associated to a vertex of the face.
\end{itemize}

From the conditions $\ell_2 u_2 = \tu_2 \tell_2$ and $\ell_6 u_6 = \tu_6\tell_6$ we find that $u_6 \tell_6^\dagger \tell_6 u_6^{-1} = \ell^\dagger_6 \ell_6$ and $\tu_2 \tell_2 \tell^\dagger_2 \tu^{-1}_2 = \ell_2 \ell^\dagger_2$, or equivalently in the vector representation
\be
R(u_6)\,\vec{\widetilde{T}}^{\operatorname{op}}_6 = \vec{T}^{\operatorname{op}}_6,\qquad R(\tu_2)\,\vec{\widetilde{T}}_2 = \vec{T}_2.
\ee
This reduces our Hamiltonian constraint to
\be\label{hamil}
H_{26} = \vec{T}_2\cdot R(u_1) \vec{T}^{\operatorname{op}}_6 - \vec{\widetilde{T}}_2\cdot \vec{\widetilde{T}}^{\operatorname{op}}_6.
\ee
The geometric interpretation of the constraint $H_{26}=0$ has been extensively discussed in \cite{HyperbolicPhaseSpace} and its content is actually the same as in the non-deformed case \cite{3DHamiltonian}. To summarize, it relates the extrinsic curvature at the edge dual to 1, and captured in $u_1$, to the intrinsic curvature. From the vectors and the holonomy $u_1$ it is possible to define a notion of dihedral angle between two hyperbolic triangles. The constraint then forces this dihedral angle to be the dihedral angle of a homogeneously curved 3D geometry, determined by the angles between the triangle edges.

%%%%%%%%%%%%%%%%%%%%%%%
\subsection{Quantization: spin 1 Hamiltonian and the Biedenharn-Elliott identity}\label{BE id}
%%%%%%%%%%%%%%%%%%%%%%%

Our aim is to quantize $H_{26}$ and impose the quantum condition $\widehat{H}_{26} |\psi\rangle =0$. The Hilbert space associated to the cell decomposition is
\be \label{HilbertSpace}
\bigoplus_{\{j_e\}} \bigotimes_v \Inv (\underbrace{j^{(*)}_{e_{1v}} \otimes j^{(*)}_{e_{2v}} \otimes \dotsb}_{\text{edges incident at $v$}} )
\ee
As we have seen the Gauss law stabilizes each fixed spin sector, which allowed to solve it for fixed spins, and selected at each vertex the invariant subspace $\Inv (j^{(*)}_{e_{1v}} \otimes j^{(*)}_{e_{2v}} \otimes \dotsb)$. The notation $j^{(*)}_e$ refers to either the space with representation $j_e$ if the edge $e$ is ingoing at $v$, or the space of the dual representation $j^*_e$ if $e$ is outgoing. Notice that the order of the edges at each vertex matters, due to the non-cocommutativity of the coproduct. After solving the Gauss law at each vertex, we take the tensor product of the invariant subspaces with fixed spins, and sum over all the possible ways to label the edges with spins.

The scalar product makes intertwiners at a given vertex with different edge labels orthogonal,
\be
( i_{j_{e_{1v}},\dotsc, j_{e_{Nv}}}, i_{j'_{e_{1v}},\dotsc, j'_{e_{Nv}}} ) = N_{j_{e_{1v}},\dotsc, j_{e_{Nv}}} \prod_{e_{iv}} \delta_{j_{e_{iv}},j'_{e_{iv}}}.
\ee
Furthermore, if the tensor product of tensor operators is not generally commutative, the tensor product of tensor operators acting on a tensor product of invariant subspaces is commutative. %This is quite welcome as it means that we do not have to use longer and longer strings of $\cR$-matrices to construct vector operators.
For example, let us consider  the tensor operator $\bt$. We know that $\bt\ot \one$ acting on $\cH\ot \cH'$ is still a tensor operator and $\one\ot \bt$ will not be in general a tensor operator. However if we restrict $\one\ot \bt$ to the invariant space generated by the intertwiners $i\in\cH$ and $i'\in\cH'$,   then $\one\ot \bt$ can be seen as a tensor operator. Hence,  unless we consider an observable that lives on at least two different intertwiners, we do not need to order the vertices \cite{ours2}.

\medskip

The Hamiltonian \eqref{hamil} contains two types of terms which have to be quantized, $\vec{\tilde{T}}_2 \cdot \vec{\tilde{T}}^{\operatorname{op}}_6$ and $\vec{T}_2\cdot R(u_1)\vec{T}^{\operatorname{op}}_6$. Let us start with the simplest one, $\vec{\tilde{T}}_2 \cdot \vec{\tilde{T}}^{\operatorname{op}}_6$. It is well-defined on $\textrm{Inv}(j_2^* \otimes j_6^* \otimes j_4)$ and acts trivially on the factor $j_4$. The eigenvalue has already been found in the Section \ref{sec:ScalarProd}. Relabeling the variables of equation \eqref{t1*t2*} correctly,
\be \label{QuantumT2T6}
\vec{\tilde{T}}_2 \cdot \vec{\tilde{T}}^{\operatorname{op}}_6\quad\rightarrow \quad \sum_{A=-1,0,1} (t_{-A} \otimes t^{\operatorname{op}}_A) \ i_{j_2^* j_6^* j_4} = N_{j_2}N_{j_6} \sqrt{[2j_2+1][2j_6+1]} (-1)^{j_2+j_6+j_4+1} \begin{Bmatrix} 1 &j_6 &j_6\\j_4 &j_2 &j_2\end{Bmatrix}\ i_{j_2^* j_6^* j_4}.
\ee
\medskip

We are left with the quantization of $\vec{T}_2\cdot R(u_1)\vec{T}^{\operatorname{op}}_6$. From the previous sections, we know the quantum versions of both vectors as well as the holonomy. By construction we know that the operator associated to $\vec{T}_2\cdot R(u_1) \vec{T}^{\operatorname{op}}_6$ should act on the invariant space $\Inv(j_1\otimes j_2\otimes j_3)\otimes \Inv(j_6\otimes j_1^*\otimes j_5^*)$ and map it to $\bigoplus_{J_1} \Inv(J_1\otimes j_2\otimes j_3)\otimes \Inv(j_6\otimes J_1^*\otimes j_5^*)$. The holonomy part $R(u_1)$ has been quantized as $\,^{q^{-1}}\cD^{*1}_{AB}(u)$, which acts on both $j_1$ and $j_1^*$, while the respective quantum versions of $\vec T_2$ and $T^{\operatorname{op}}_6$, $\bt_2$ and $\bt^{\operatorname{op}}_6$, respectively act on $j_2$ and $j_6$. Using the decomposition $\,^{q^{-1}}\cD^{*1}_{AB}(u) = \sum_{J_1} D^{\operatorname{op}}_A\ot D^{\operatorname{op}*}_B$, obtained in  \eqref{QuantumHol}, we can construct a scalar operator by combining in adequate manner the different vector operators,
%are now ready to quantize  the  invariant part of $H_{26}$ containing the holonomy, i.e. $\vec{T}_2\cdot R(u_1) \vec{T}^{\operatorname{op}}_6$ where
%$R(u_1)$ is contracted with $\vec{T}_2, \vec{T}^{\operatorname{op}}_6$ to get an invariant. Upon quantization we get
\be \label{QuantizedHolomonieHamilton}
\vec{T}_2\cdot R(u_1) \vec{T}^{\operatorname{op}}_6 = \sum_{A,B} (T_2)_{A} R^{1}_{AB}(u_1) (T_6^{\operatorname{op}})_{-B}
 \rightarrow\ \sum_{A}(D^{\operatorname{op}}_A\otimes t_{A}\otimes \one) \otimes \sum_B(t^{\operatorname{op}}_{-B} \otimes D_B^{\operatorname{op}*} \otimes \one).
\ee
%We note there is no $\cR$-matrices entering in the game, which makes things simpler.
The part $\sum_A D^{\operatorname{op}}_A\otimes t_{A}\otimes \one$ is well-defined as a map $\Inv(j_1\otimes j_2\otimes j_3)\to \Inv(J_1\otimes j_2\otimes j_3)$, while $\sum_B t^{\operatorname{op}}_{-B} \otimes D_B^{\operatorname{op}*} \otimes \one$ is well-defined as $\Inv(j_6\otimes j_1^*\otimes j_5^*)\to \Inv(j_6\otimes J_1^*\otimes j_5^*)$ and they both send invariant states to invariant states, by construction. It is thus possible to evaluate their action separately. We start with
\be
\begin{aligned}
\sum_A &(D^{\operatorname{op}}_A\otimes t_{A}\otimes \one)\ i_{j_1 j_2 j_3} \\
&= N_{j_2} \sum_{m_1,m_2,A} (-1)^{J_1-j_1+1}\sqrt{\frac{[2J_1+1]}{[2j_1+1]}}\,C^{1 J_1 j_1}_{-A -M_1 -m_1} C^{1 j_2 j_2}_{A m_2 n_2} C^{j_1 j_2 j_3}_{m_1 m_2 -m_3} (-1)^{j_3-m_3} q^{-\frac{m_3}{2}} |J_1 M_1, j_2 n_2, j_3 m_3\rangle\\
&= N_{j_2} (-1)^{j_1+j_2+j_3} \sqrt{[2J_1+1][2j_2+1]} \begin{Bmatrix} J_1 &j_1 &1\\ j_2 &j_2 &j_3\end{Bmatrix}\ i_{J_1 j_2 j_3},
\end{aligned}
\ee
and then
\be
\begin{aligned}
\sum_B &(t^{\operatorname{op}}_{-B}\otimes D_B^{\operatorname{op}*}\otimes \one)\ i_{j_6 j_1^* j_5^*}\\
&= N_{j_6} \sum_{m_6,n_1,B} C^{1 j_6 j_6}_{-B -n_6 -m_6} (-1)^{1-B} q^{\frac{B}{2}} C^{1 j_1 J_1}_{B -n_1 -N_1} (-1)^{j_1-n_1} q^{-\frac{n_1}{2}} C^{j_6 j_1 j_5}_{m_6 -n_1 m_5} |j_6 n_6\rangle \otimes \langle J_1 N_1|\otimes \langle j_5 m_5| \\
&= N_{j_6} (-1)^{J_1+j_5+j_6} \sqrt{[2J_1+1][2j_6+1]} \begin{Bmatrix} J_1 &j_1 &1\\ j_6 &j_6 &j_5\end{Bmatrix}\ i_{j_6 J_1^* j_5^*}.
\end{aligned}
\ee
We re-assemble the pieces to get the quantized version of $\vec{T}_2\cdot R(u_1) \vec{T}^{\operatorname{op}}_6$,
\begin{multline} \label{QuantumT2U1T6}
\sum_{A,B} \left[(D_A\otimes t_{A}\otimes \one) \otimes (t^{\operatorname{op}}_{-B} \otimes D_B^* \otimes \one)\right]\ i_{j_1 j_2 j_3}\otimes i_{j_6 j_1^* j_5^*} \\
= N_{j_2} N_{j_6} \sqrt{[2j_2+1][2j_6+1]} \sum_{J_1} (-1)^{J_1+j_1+j_2+j_3+j_5+j_6} [2J_1+1] \begin{Bmatrix} J_1 &j_1 &1\\ j_2 &j_2 &j_3\end{Bmatrix} \begin{Bmatrix} J_1 &j_1 &1\\ j_6 &j_6 &j_5\end{Bmatrix}\ i_{J_1 j_2 j_3}\otimes i_{j_6 J_1^* j_5^*}
\end{multline}

\medskip

We have now achieved the quantization of all the pieces required to get the quantum Hamiltonian constraint in  \eqref{QuantumT2T6} and \eqref{QuantumT2U1T6}. It allows to write down the corresponding Wheeler-DeWitt equation explicitly. A state can be expanded like
\be
\psi = \sum_{\{j_e\}} \psi(j_1,\dotsc,j_6,\dotsc)\,\prod_e [2j_e+1]\ i_{j_1 j_2 j_3}\otimes i_{j_6 j_1^* j_5^*}\otimes i_{j_2^* j_6^* j_4} \bigotimes_{v'} i_{v'}
\ee
where the tensor product over $v'$ indicates all the vertices of the cell decomposition others than $v_{123}, v_{615}, v_{264}$ and the corresponding intertwiners ensuring the local Gauss law. To avoid lengthy formulas, we introduce the following notations,
\begin{align}
A_{\pm}(j_1) &= \sum_{J_1} \delta_{J_1,j_1\pm1} (-1)^{J_1+j_1+j_2+j_3+j_5+j_6} [2J_1+1] \begin{Bmatrix} J_1 &j_1 &1\\ j_2 &j_2 &j_3\end{Bmatrix} \begin{Bmatrix} J_1 &j_1 &1\\ j_6 &j_6 &j_5\end{Bmatrix}\\
A_0(j_1) &= (-1)^{2j_1+j_2+j_3+j_5+j_6} [2j_1+1] \begin{Bmatrix} j_1 &j_1 &1\\ j_2 &j_2 &j_3\end{Bmatrix} \begin{Bmatrix} j_1 &j_1 &1\\ j_6 &j_6 &j_5\end{Bmatrix} - (-1)^{j_2+j_6+j_4+1} \begin{Bmatrix} j_2 &j_2 &1\\ j_6 &j_6 &j_4\end{Bmatrix}
\end{align}
$A_\pm(j_1)$ is the coefficient of the term with $J_1=j_1\pm1$ in \eqref{QuantumT2U1T6}. $A_0(j_1)$ is the coefficient of the term with $J_1=j_1$ in \eqref{QuantumT2U1T6} minus the eigenvalue of \eqref{QuantumT2T6}. Acting with the quantum constraint on $\psi$ gives
\begin{multline}
\widehat{H}_{26} \psi = \sum_{\{j_e\}} \psi(j_1,\dotsc)\ N_{j_2}N_{j_6} \prod_{e} [2j_e+1]\\
\bigl(A_-(j_1) i_{j_1-1 j_2 j_3}\otimes i_{j_6 (j_1-1)^* j_5^*} + A_0(j_1) i_{j_1 j_2 j_3}\otimes i_{j_6 j_1^* j_5^*} + A_+(j_1) i_{j_1+1 j_2 j_3}\otimes i_{j_6 (j_1+1)^* j_5^*}\bigr)
 \otimes i_{j_2^* j_6^* j_4} \bigotimes_{v'} i_{v'}
\end{multline}
To relabel the sum over $j_1$, we notice that
\be
A_{\pm}(j_1\mp1) = \frac{[2j_1+1]}{[2(j_1\mp1)+1]}\ A_{\mp}(j_1),
\ee
and therefore
\begin{multline}
\widehat{H}_{26} \psi = \sum_{\{j_e\}} N_{j_2}N_{j_6} \bigl(A_-(j_1) \psi(j_1-1,j_2,\dotsc) + A_0(j_1) \psi(j_1,j_2\dotsc) + A_+(j_1) \psi(j_1+1,j_2\dotsc)\bigr)\\
\prod_{e} [2j_e+1]\ i_{j_1 j_2 j_3}\otimes i_{j_6 j_1^* j_5^*} \otimes i_{j_2^* j_6^* j_4} \bigotimes_{v'} i_{v'}.
\end{multline}
Thanks to the orthogonality of intertwiners with different spin labels, the constraint $\widehat{H}_{26} \psi = 0$ leads to a recursion on the coefficients of the expansion $\psi(j_1,j_2,\dotsc)$,
\be \label{Recursion}
A_-(j_1) \psi(j_1-1,j_2,\dotsc) + A_0(j_1) \psi(j_1,j_2\dotsc) + A_+(j_1) \psi(j_1+1,j_2\dotsc) = 0.
\ee
It is of order 2 since it involves three consecutive neighbours. However, a single initial condition suffices to implement it. Indeed, on the boundary of the domain satisfying the triangle inequalities, i.e. $j_1^{\operatorname{min}} = \max(|j_2-j_3|,|j_5-j_6|)$, the lowering coefficient $A_{-}(j_1^{\operatorname{min}})$ vanishes, making the recursion first order. Moreover, up to the choice of the initial condition (depending on $j_2,\dots$), the solution is known to be the $q$-6j symbol. Let us define the function $\phi(j_2,\dotsc)$ via
\be
\psi(j_1^{\operatorname{min}},j_2,\dotsc) = \begin{Bmatrix} j_1^{\operatorname{min}} &j_2 &j_3\\ j_4 &j_5 &j_6\end{Bmatrix}\ \phi(j_2,\dotsc),
\ee
then the recursion generates
\be
\psi(j_1,j_2,\dotsc) = \begin{Bmatrix} j_1 &j_2 &j_3\\ j_4 &j_5 &j_6\end{Bmatrix}\ \phi(j_2,\dotsc).
\ee
Moreover, there is a Hamiltonian constraint for each vertex of the face, meaning that the same recursion holds on $j_2$ and on $j_6$. It is found implementing them that $\phi$ depends on all the spins but $j_1,j_2,j_6$. The dependence of the state on $j_1,j_2,j_6$ is entirely captured in a 6j-symbol,
\be \label{6jFactorization}
\psi(j_1,j_2,j_3,j_4,j_5,j_6,\dotsc) = \begin{Bmatrix} j_1 &j_2 &j_3\\ j_4 &j_5 &j_6\end{Bmatrix}\ \phi(j_3,j_4,j_5,\dotsc).
\ee

%%%%%%%%%%%%%%%%%%%%%%%
\subsection{Towards spin foam transition amplitudes of the Turaev-Viro type}
%%%%%%%%%%%%%%%%%%%%%%%

Since the full dependence of the physical state on the spins along the boundary of the face has been extracted, the function $\phi$ should be determined by a set of constraints on the graph $G/f_{126}$ where the face and its boundary edges 1,2,6 have been removed (or rather continuously shrunk to a point so as to preserve the topology), as depicted in the Figure \ref{fig:RibbonMove}. This is  expected from the classical theory itself. Indeed, one can gauge fix the local $\SU(2)$ symmetry so as to set $\tilde u_2=u_6=\one$. Then the flatness constraint on the face simply reduces to $u_1=\one$. That in turn enforces $\tu_1=u_2=\tu_6=\one$ as well as $\ell_1=\tell_1, \ell_2=\tell_2, \ell_6=\tell_6$. One is left with the three Gauss laws on the vertices $v_{123}, v_{615}, v_{264}$ from which a condition on the external variables of the edges 3,4,5 only can be extracted,
\be
\left.\begin{aligned}
&\ell_1\ell_2\ell_3=\one\\
&\ell_2^{-1} \ell_6^{-1} \ell_4 =\one\\
&\ell_6 \ell_1^{-1} \tell_5^{-1} = \one
\end{aligned} \right\} \quad \Rightarrow\qquad \tell_5^{-1}\,\ell_4\,\ell_3 = \one.
\ee
That means that the face and its boundary edges is replaced with a single (ribbon) vertex with the Gauss law $\tell_5^{-1} \ell_4 \ell_3 =\one$.

\begin{figure}
\includegraphics[scale=.4]{3ValentFace.pdf} \hspace{3cm} \includegraphics[scale=.5]{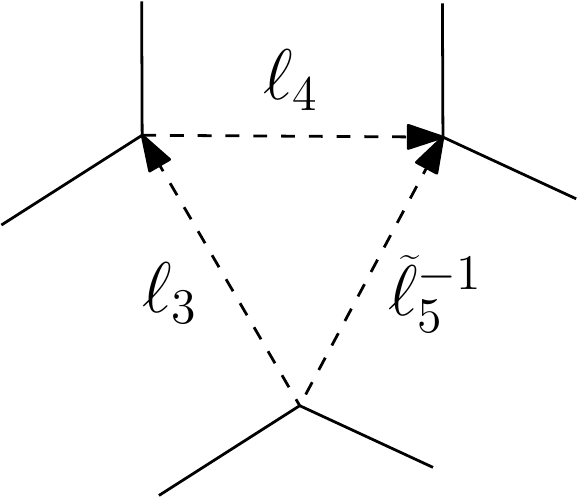}
\caption{\label{fig:RibbonMove} On the left: the face $f_{126}$ of the graph $G$ on which we solve the quantum flatness constraint. On the right: the same region in the graph $G/f_{126}$.}
\end{figure}

Even if we have solved the Gauss laws on $v_{123}, v_{615}, v_{264}$ and the flatness constraint on the face $f_{126}$ to get the factorization of the physical state \eqref{6jFactorization}, the properties of $\psi$ under the quantum version of $\tell_5^{-1} \ell_4 \ell_3$ is not clear to us yet.

The change of graph from $G$ to $G/f_{126}$ is a 3-to-1 move, as can be seen in the dual picture. The initial situation, with the face $f_{126}$, dually corresponds to three triangles glued two by two around a vertex of degree 3 (which is dual to the face of degree 3). The edges 1,2,6 are dual to the internal edges, and the edges 3,4,5 are dual to the exterior edges of the gluing. The move, contracting the face and its boundary edges, dually consists of removing the internal edges and keeping only the external edges which then form a single triangle. It is depicted in the Figure \ref{fig:Move1To3}.

\begin{figure}
\includegraphics[scale=.7]{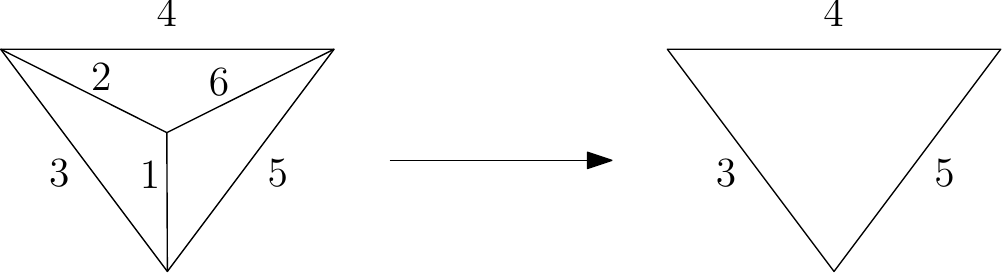}
\caption{\label{fig:Move1To3} On the left, the dual to the left of the Figure \ref{fig:RibbonMove} and the right the dual to the right of the Figure \ref{fig:RibbonMove}.}
\end{figure}

The transition from a state $\psi(j_1,\dotsc)$ on $G$ to a state $\phi(j_3,j_4,j_5,\dotsc)$ on $G/f_{126}$ has to preserve the spins of the common edges. Moreover the result \eqref{6jFactorization} is the first step to show that the amplitude associated to this transition in the spin representation is a $q$-6j symbol. This would be a direct $q$-deformation extension of the already well-known result in the flat case at $q=1$, \cite{3DHamiltonian, ScalarProd3D},
\be
\left(\begin{array}{c} \includegraphics[scale=.4]{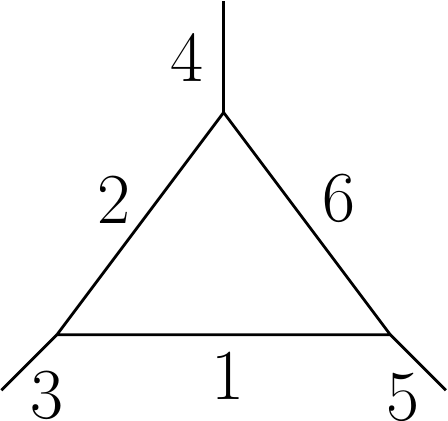} \end{array} , \begin{array}{c} \includegraphics[scale=.4]{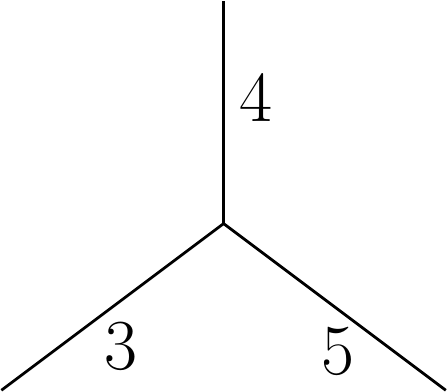} \end{array} \right) \propto \begin{Bmatrix} j_1 &j_2 &j_3\\ j_4 &j_5 &j_6\end{Bmatrix}\ \text{$\times$ rest}.
\ee
(The sign $\propto$ has to do with the normalization of the intertwiners which we have not cared about throughout the paper but which obviously play a role in the final physical scalar product between states. Generalizing the $q=1$ case, those normalizations are products and ratios of quantum dimensions.)

%%%%%%%%%%%%%%%%%%%%%%
\section{Conclusion}
%%%%%%%%%%%%%%%%%%%%%%

%%%%%%%%%%
\subsection{Summary}
%%%%%%%%%%

This paper is the beginning of the quantization of the classical system proposed in \cite{HyperbolicPhaseSpace} to describe homogeneously curved discrete geometries (negative cosmological constant). The classical phase space is based on $\SL(2,\C)$ viewed as a deformation of $T^*\SU(2)$.
\begin{itemize}
\item We have quantized the deformed momenta (which generalizes the fluxes of loop quantum gravity), which consist of the lower triangular matrix $\ell$ appearing in the Iwasawa decomposition $G=\ell u$ with $u\in\SU(2)$. Their quantum commutation relations generate the quantum algebra $\UQ$.
\item The Gauss law which generates local $\SU(2)$ transformations is imposed at the quantum level, where it is shown to enforce local $\UQ$ invariance, at each vertex, with the expected coproduct. Solutions in the spin representation are $q$-deformed intertwiners.
\end{itemize}
Those results are direct generalizations of the kinematics of loop quantum gravity.

Moreover the classical quantities $\vec{T} = \tr \ell\ell^\dagger \vec{\sigma}$ and $\vec{T}^{\operatorname{op}} = \tr \ell^\dagger\ell\vec{\sigma}$ which transform as 3-vectors are turned into vectors of operators, satisfying some specific equivariance properties \eqref{TransfoVectorOperator}, \eqref{EquivarianceTop}. They are such that the scalar product of those vectors of operators are $\UQ$ invariant and can thus be used to define observables. This result
\begin{itemize}
\item makes the connection with the framework introduced in \cite{ours1, ours2} based on tensor operators,
\item allows to investigate the quantization of the Hamiltonian constraint.
\end{itemize}
We have further solved this constraint explicitly on a face of degree three, in the spin representation, and find that it enforces a factorization of the physical state into a $q$-6j symbol times a state which satisfies the constraints on a lattice where the face has been shrunk. This is precisely a realization of the 2D Pachner move 3-to-1, a first step towards spin foam transition amplitudes for real $q$ and quantum coarse graining.

%%%%%%%%%%
\subsection{Open questions for future investigations}
%%%%%%%%%%

All the results presented here turn out to be $q$-deformed extensions of results which are well-known in 3D gravity with vanishing cosmological constant (on the lattice). However, we are not quite yet at the same level of understanding in the deformed case and many questions remain to be answered.

\smallskip

\paragraph*{\textbf{Pachner moves}:}
For instance, we have not completed the quantum analysis of the 3-to-1 move as it requires a better understanding of the behavior of $\UQ$ transformations under coarse graining. The Pachner move 2-to-2 also has to be analyzed ; the associated transition amplitude is expected to be a $q$-6j symbol again.

\smallskip

\paragraph*{\textbf{Quantization procedure}:} We have been able to quantize the key elements of the theory, the Gauss constraint and the Hamiltonian constraint, however much is left for a complete understanding of the quantization procedure. %For example, it is not clear how to construct observables acting on the third leg of an intertwiner.
  The observables introduced in \cite{ours1, ours2} are made out of tensor operators, conjugated with $\cR$-matrices to ensure the correct braiding. In the present work, it is  remarkable that \textit{no} $\cR$-matrix is needed in the construction of observables (note that we restricted our attention to observables acting on the two first legs of the intertwiner). Instead we deal with two different kinds of vector operators: the standard one  $\bt$  providing the quantization of $\vec{T}$, and $\bt^{\operatorname{op}}$, the quantization of $\vec{T}^{\operatorname{op}}$. This is new since when the deformation parameter $\kappa$ goes to zero, both vectors and both operators coincide. However, for $\kappa\neq 0$, $\vec{T}^{\operatorname{op}}$ and $\vec{T}$ are related through a $\SU(2)$ transformation $h$ (coming from the Cartan decomposition $\ell =Bh$ where $B$ is a boost). Following the result in \eqref{cool}, it seems that this $h$ transformation could be related to the $\cR$-matrix. This point illustrates that we do not have  yet a complete understanding of the quantization procedure.

\smallskip

\paragraph*{\textbf{Deformed spinors}:} A modern tool for quantum models of geometry on the lattice is the spinor formalism \cite{twisted, un0, un1, twisted1}. Spinor variables (living on the fundamental representation of $\SU(2)$) indeed enable to generate all observables in any representations, to reconstruct the holonomies, to evaluate generating functions of spin network evaluations \cite{Jeff, CostantinoMarche, GeneratingFunctions} and it has been key to discovering the $\U(N)$ symmetry of the observable algebra at a vertex of the graph \cite{un0,un1}. It is a framework that encompasses all the results obtained so far.  It would be very useful to develop such a formalism in the deformed case as well.  The first step, that is the definition of the deformed classical spinors,  has been described  in \cite{spin}.
Their  quantum part has  been developed in \cite{ours2} using spinor operators. Spinors would in principle simplify the quantization of all the dynamical variables (such as the boost $B$ and the rotation $h$) and allow to represent them as operators carrying arbitrary representations. In the present paper for instance, the Hamiltonian constraint (which is invariant) is built from operators carrying the spin 1 (the vector operators and the holonomy in the vector representation). It therefore generates a recursion on the physical state with shifts of the spin labels by $\pm1$. It means that families of integer spins and half-integer spins do not talk to each other at this point (and both require an initial condition). To really solve the model at once, it is required to define the Hamiltonian using the holonomy in the representation of spin 1/2, as it has been done using spinors in the case $q=1$ \cite{Spin1/2Hamiltonian}.

\smallskip

\paragraph*{\textbf{4D case}:}
Eventually, since the form of the Hamiltonian constraint we have used here is the same as the one in the flat case in 3D \emph{and in 4D} \cite{Recursion4D}, there might be a way to extend our results, considering that the Hamiltonian generate a recursion on the physical state in the spin basis, to the homogeneously curved, and still topological, sector (of the $BF +\Lambda B^2$ type) in 4D.

\section*{Acknowledgement}
M. Dupuis and F. Girelli acknowledge financial support from the Government of Canada through
respectively a Banting fellowship and a NSERC Discovery grant.

%%%%%%%%%%%%%
\appendix
%%%%%%%%%%%%%%%%%%%%%%%%%%%%%
\section{Definitions and Notations on $\UQ$}\label{uq}
%%%%%%%%%%%%%%%%%%%%%%%%%%%%%%

The $q$-deformation of the universal enveloping algebra of $\SU(2)$, $\UQ$ is generated by $J_\pm, J_z$ which satisfy the commutation relations
\be \label{commutation}
[J_z, \, J_\pm ]=\pm J_\pm, \qquad [J_+, \, J_-]=\f{q^{J_z}-q^{-J_z}}{q^{\f12}-q^{-\f12}}.
\ee
Setting $K\equiv q^{\f{J_z}{2}}$, they are equivalent to
\be \label{commutation1}
K J_\pm K^{-1}=q^{\pm \f12} J_\pm, \qquad [J_+, \,J_-]=\f{K^2-K^{-2}}{q^{\f12}-q^{-\f12}}.
\ee

The co-algebra structure of $\UQ$ is defined by
\begin{itemize}
\item a coproduct,
\be \label{coproduct}
\begin{aligned}
\Delta:\ &\UQ \to \UQ \otimes \UQ \\
&J_\pm \mapsto \Delta(J_\pm) = J_\pm \otimes K + K^{-1} \otimes J_\pm\\
&K\mapsto \Delta(K)=K\otimes K,
\end{aligned}
\ee
\item an antipode: $S(J_\pm)=-q^{\pm \f12} J_\pm, \; S(K)=K^{-1}$;
\item a counit: $\epsilon(J_\pm)=0, \; \epsilon(K)=1$.
\end{itemize}

The representation theory of $\UQ$ is very similar to the one of $\su(2)$, except that $q$-numbers are now entering into the game. We introduce the traditional $q$-number notation,
\be
[n]\equiv \f{q^{\f{n}{2}}-q^{-\f{n}{2}}}{q^{\f12}-q^{-\f12}}.
\ee
The irreducible representations of $\UQ$ are labeled by a non-negative half-integer and and act on Hilbert spaces of dimension $2j+1$, $j\in\N/2$. The states diagonalizing $K$ form an orthonormal basis $(|j,-j\ra,|j,-j+1\ra,\dotsc,|j,j-1\ra, |j,j\ra)$ and
\be \label{irreps}
K|j\,m\ra = q^{\frac{m}{2}} |j \, m\ra,\qquad \text{and}\qquad J_\pm |j\, m\ra= \sqrt{[j\mp m][j\pm m +1]} |j\, m \pm 1 \ra,\qquad \text{for $m=-j,\dotsc,j$}.
\ee

The CG coefficients map the module $j_1\otimes j_2$ to $j_3$, meaning that the action of the generators has to commute with map. This gives a recursion on the coefficients,
\begin{multline} \label{CGRecursion}
q^{\frac{m_2}{2}}\sqrt{[j_1\pm m_1][j_1\mp m_1+1]} C^{j_1 j_2 j_3}_{m_1\mp1 m_2 m_3} + q^{-\frac{m_1}{2}}\sqrt{[j_2\pm m_2][j_2\mp m_2+1]} C^{j_1 j_2 j_3}_{m_1 m_2\mp1 m_3} \\
= \sqrt{[j_3\mp m_3][j_3\pm m_3+1]} C^{j_1 j_2 j_3}_{m_1 m_2 m_3\pm1}
\end{multline}
The CG coefficients enjoy some symmetries that are used quite often throughout the main text,
\begin{align}
C^{j_1 j_2 j_3}_{m_1 m_2 m_3} &= C^{j_2 j_1 j_3}_{-m_2 -m_1 -m_3}\\
&= (-1)^{j_1-m_1} q^{\frac{m_1}{2}} \sqrt{\frac{[2j_3+1]}{[2j_2+1]}} (-1)^{j_1+j_3-j_2} C^{j_1 j_3 j_2}_{-m_1 m_3 m_2}\\
&= (-1)^{j_2+m_2} q^{-\frac{m_2}{2}} \sqrt{\frac{[2j_3+1]}{[2j_1+1]}} (-1)^{j_2+j_3-j_1} C^{j_3 j_2 j_1}_{m_3 -m_2 m_1}.
\end{align}
The $q$-6j symbol satisfies
\begin{multline}
\sum_{m_2, m_3, m_4} C^{j_1 j_2 j_3}_{m_1 m_2 m_3} (-1)^{j_2-m_2} q^{\frac{m_2}{2}} C^{j_2 j_6 j_4}_{-m_2 m_6 m_4} C^{j_3 j_4 j_5}_{m_3 m_4 m_5} \\= \begin{Bmatrix} j_1 &j_2 &j_3 \\ j_4 &j_5 &j_6\end{Bmatrix} (-1)^{j_1+j_2+j_4+j_5} (-1)^{j_2+j_6-j_4} \sqrt{[2j_3+1][2j_4+1]} C^{j_1 j_6 j_5}_{m_1 m_6 m_5}.
\end{multline}
%%%%%%%%%%%%%%%%%%%%%%%%%%

\end{document}